\documentclass[11pt,a4paper]{article}

\usepackage[margin=1in]{geometry}
\usepackage{amsmath}
\usepackage{amssymb}
\usepackage{amsthm}
\usepackage{array}
\usepackage{booktabs}
\usepackage{needspace}
\usepackage{listings}
\usepackage{pgfplots}
\pgfplotsset{compat=1.17}
\usetikzlibrary{decorations.pathreplacing}
\lstdefinestyle{walk}{basicstyle=\small\ttfamily, mathescape=true, columns=fullflexible,
  keepspaces=true, frame=single, framesep=6pt, xleftmargin=2pt, xrightmargin=2pt}
\usepackage[numbers]{natbib}
\usepackage[hidelinks,unicode]{hyperref}

\newcommand{\doi}[1]{\href{https://doi.org/#1}{\nolinkurl{doi:#1}}}

\hypersetup{
  pdftitle={Certified Panic Mode: Repair-Invariant Error Recovery for Maximal-Munch Lexing},
  pdfauthor={Nicklas Nidhögg},
  pdfsubject={Lexical error recovery; deterministic finite automata; maximal munch tokenization},
  pdfkeywords={error recovery, panic mode, tokenization, deterministic finite automata, maximal munch}
}

\theoremstyle{definition}
\newtheorem{definition}{Definition}

\theoremstyle{plain}
\newtheorem{lemma}{Lemma}
\newtheorem{theorem}{Theorem}
\newtheorem{corollary}{Corollary}
\newtheorem{proposition}{Proposition}

\theoremstyle{remark}
\newtheorem{remark}{Remark}

\newcommand{\code}[1]{\texttt{#1}}
\newcommand{\Rreach}{R_{\mathrm{reach}}}
\newcommand{\Rcomp}{R_{\mathrm{complete}}}
\newcommand{\suffix}[2]{#1[#2..]}
\newcommand{\prefix}[2]{#1[0..#2)}
\newcommand{\window}[3]{#1[#2..#3)}
\newcommand{\con}{\operatorname{con}}

\newcommand{\ExactBasisBound}{8}
\newcommand{\ExactFamilyBrute}{12}
\newcommand{\ExactFamilyMax}{64}
\newcommand{\ExactInstances}{977}
\newcommand{\ExactQueries}{17{,}358}
\newcommand{\ExactRecurrenceCells}{4{,}096}
\newcommand{\ExactScans}{256{,}088}
\newcommand{\ExactShippedSteps}{2{,}098{,}176}
\newcommand{\ExactUnaryLength}{2{,}048}


\title{Certified Panic Mode:\\Repair-Invariant Error Recovery for Maximal-Munch Lexing}
\author{Nicklas Nidh\"ogg\\[0.4em] \normalsize Independent Researcher\\ \normalsize
\texttt{nicklas.nidhogg@gmail.com}\\ \normalsize ORCID: 0009-0006-6161-6150}
\date{September 6, 2026}

\begin{document}

\maketitle

\begin{abstract}
Classical panic-mode recovery skips a failed scan forward by convention, promising progress but no repair-universal
boundary guarantee for where it lands. We choose the resume position by theorem, via a committed-prefix lemma over
certificates posted for parallel lexing (arXiv:2608.03473, arXiv:2608.09761): a position is a sound recovery point when
every prefix repair whose scan commits through the certificate's returned evidence places a token boundary there. This
quantifier is strictly stronger than its complete-repair special case, and at a sound recovery point the inclusion is
strict precisely when no repair completes while some scan commits through the evidence. The evidence-returning interface
pairs each answer with the interval it rests on, so a caller who knows what the scanner cannot check decides whether the
boundary transfers to the intended input. The search provably terminates; the guarantee is relative to the active mode's
automaton. For an automaton of $n$ states the quantifier is decided exactly from the preserved tail by at most $n + 1$
representative repairs, and a shortest reaching repair has at most $n - 1$ symbols and a shortest counterrepair at most
$n$, both tight. Every anchor of an $m$-byte input is profiled in $O(nm + n|\Sigma|)$ time. The shipped repair routines
are the complete-repair case, which proves their contracts. In a deterministic corruption study of munch's procedure
beside skip-one and delimiter recovery, every certified answer whose evidence survives the damage passes its landing
check on every move, $40{,}885$ of $40{,}885$, while the $5{,}595$ first answers resting on evidence the damage reached
land $1{,}589$ times, reported, not asserted. Of the $16{,}808$ trials labeled beyond repair, $13{,}522$ resume on a
suffix that tokenizes whole, where both quantifiers are provably empty, and $3{,}286$ stay undetermined, their damaged
inputs not archived.
\end{abstract}

\section{Introduction}

Two companion papers certify positions in completely tokenizable input. A certified byte begins a token at every
occurrence in every completely tokenizable input, and a certified split window places a token origin at a fixed offset
into every occurrence~\citep{splitpoints2026, splitwindows2026}. Error recovery lives where those hypotheses fail, in
input that is not completely tokenizable. The classical answer is panic mode, which works at two levels. The lexical
form deletes characters until a well-formed token appears, and the parser form discards input up to a designated
synchronizing delimiter such as a newline or a semicolon. Both are old and widespread. They guarantee progress and,
where a synchronizer is found, membership in the designated set, but neither supplies a repair-universal boundary
guarantee for the position it lands on. This study moves the delimiter convention down a level and evaluates it, in both
placements, beside the lexical form. This paper supplies the theorem the convention never had. It shows that the same
static certificates that plan parallel chunk boundaries in valid input certify resynchronization points in broken input,
with a guarantee quantified over every prefix repair whose scan commits through the certificate's preserved evidence.

Three naive statements fail, and their failures shape the definition. First, demanding that the resumed suffix tokenize
on its own is reasonable but insufficient. Over $\{\code{ab}, \code{b}\}$ the suffix \code{b} tokenizes alone, yet the
repair \code{a} produces the whole \code{ab}, one token with no boundary at the cut, so standalone suffix tokenizability
filters positions without aligning any of them with a repaired whole. Second, a certified window with positive origin
spans bytes before the position it certifies, so the resumed suffix alone does not contain the occurrence the
certificate speaks about; a statement quantifying only over the suffix cannot use the window theorem. Third, the prefix
guarantee of the split-points paper speaks about the tokens committed before a failure, never about how to continue past
it. The definition that survives all three is repair invariance: a position is a sound recovery point when every repair
of the broken prefix whose scan commits through the certificate's evidence places a committed token boundary there, with
complete tokenizability as the corollary case. It quantifies over repaired inputs, which defeats the first failure; the
theorems' conditions on the repair anchor, the left edge of the text a repair keeps, pin the certificate's evidence
inside the unmodified region, which defeats the second; and it is a statement about continuation past the failure, not
about the tokens committed before it, as the third failure requires.

The resynchronization results are one-line consequences of the posted certificates, and intentionally so. Nothing new
needs to be proved about the automaton to state them, only the right question asked of theorems that already exist. The
contribution is the recovery contract, its operational validation, and an exact decision procedure for the quantifier
the contract introduces. Certified bytes are repair-invariant at their own position. Certified windows are
repair-invariant at their occurrence plus the certified origin under one extra constraint on the repair, that it leaves
the occurrence intact. The searching algorithm satisfies this constraint by construction because it only reports
evidence lying wholly at or after its own search anchor. Whether the text past that anchor is truly unmodified is the
caller's trust decision, and the evaluation prices it. A progress lemma makes the scan-and-recover loop terminate, and
mode-driven scanning gets a per-automaton certificate, relative to the active mode, never a theorem over repairs that
change modes. The procedure ships in the munch lexing library~\citep{munch}, and the evaluation measures what the
theorem is worth against the conventions it replaces, in quality.

The contributions are:
\begin{itemize}
\item Repair-invariant resynchronization, a definition of recovery-point soundness that quantifies over every prefix
repair whose scan commits through the certificate's evidence, complete repairs the special case, with the three naive
alternatives shown to fail (Section~\ref{sec:repair}).
\item Two theorems, one per certificate kind, each a short consequence of the companions' posted results, plus the
progress lemma and the mode-scope remark that state a driver's obligations (Section~\ref{sec:repair}).
\item An exact decision procedure for the evidence-reaching quantifier, at most $n + 1$ representative repairs deciding
inhabitation and invariance at every evidence endpoint. A shortest reaching repair has at most $n - 1$ symbols, and a
shortest counterrepair at most $n$, or $n - 1$ for a position past the anchor, every bound tight. A single threshold
governs nonvacuous invariance, two repairs settle an endpoint and three settle every endpoint, and every anchor of an
input is profiled in $O(nm + n|\Sigma|)$ time and $O(n + m)$ space. The shipped repair routines are the complete-repair
case, which proves their contracts (Section~\ref{sec:exact}).
\item The searching primitive and its driver as shipped, with the one divergence from the planners' certificate policy
explained, and the fail-closed fixtures that pin each layer (Sections~\ref{sec:search} and~\ref{sec:impl}).
\item A deterministic corruption study whose oracle is exact for the theorem-predicted boundary consequence on the
exhibited pristine repair. The pristine corpus supplies the repair the theorems quantify over, so the landing check runs
on every returned answer whose supporting evidence the damage spares. The library's anchored complete-repair decider,
whose contract Section~\ref{sec:exact} proves, is evaluated under that contract, in blind and oracle-anchored forms,
beside the walk as a comparator for a different property. Repairability stratifies every trial, and the classical
skip-one and skip-to-delimiter conventions are measured in the same campaign (Section~\ref{sec:eval}).
\end{itemize}

The paper separates three questions a damaged scan raises: whether any repair completes the input, whether a returned
position is a boundary under every repair whose scan commits through its evidence, and how often, under measured damage,
that evidence survives. Certificates from the companion papers identify positions, and the resynchronization theorems
transfer their guarantee to every evidence-reaching repair of a failed scan's prefix. The collapse proposition and its
dichotomy corollary locate exactly where the strengthened quantifier adds content, and Section~\ref{sec:exact} decides
from the preserved tail which side of that line any anchor falls on. The overhang lemma tells a caller when the evidence
is forced by arithmetic alone to be \emph{covered}, beginning at or past the end of the damage, and the campaign
measures availability, coverage, and displacement where the theorems' premises hold and where they fail.

\section{Preliminaries}

The model is the companions', restated and not reproved. A token set compiles to a deterministic finite automaton
scanned by the usual longest-match loop. From an offset the scanner consumes bytes while transitions are defined,
remembers the last accepting position, emits the token found there, and restarts at that position in the initial state.
Emitted tokens have positive length. The committed position strictly advances. An accepting initial state (a nullable
pattern's minimization) never emits at zero width, and a scan that cannot advance fails. Write $\mathrm{con}(x)$ for the
scanner's final committed offset on input $x$, and call $x$ completely tokenizable when $\mathrm{con}(x) = |x|$. The
\emph{failure offset} of a failed scan is its final committed offset $f = \mathrm{con}(x)$. It is the start of the token
attempt that failed, not the read head where the transition died, which lies at or past $f$ and equals it exactly when
the very first byte of the attempt has no transition. For a scan resumed at offset $s$, the failure offset is the
absolute $f_s = s + \mathrm{con}(x[s..])$, the driver's current committed offset. The shipped driver always searches
from one past the failure offset. The evaluation's oracle-aided arm and the shipped clean-anchored interface name
another anchor, and they do so by calling the primitive directly, never through the failure-anchored driver.

A byte is \emph{formally certified} when the split-points paper's static condition holds of it. That paper's
characterization theorem makes the condition exact, in that a byte is certified if and only if every occurrence of it,
in every completely tokenizable input, begins a token. The characterization carries that paper's standing assumption
that the initial state is itself live. The condition is vacuously true of a byte no completely tokenizable input
contains at all, and the library reports only the \emph{useful} certified bytes, those with a live initial-state
transition, which is the class the split-points paper reports. Throughout this paper \emph{certified byte} means the
reported, useful kind, narrower than the companions' usage, in which a certified byte may be vacuous. The walk, the
worked example, the refusal claims, and the evaluation all speak of what the search reports, and the vacuous case is
named where it matters.

A \emph{certified window} $(W, o)$ is a nonempty byte string with a certified origin $o$, $0 \le o < |W|$: in every
completely tokenizable input containing an occurrence of $W$, the token covering the occurrence's final byte begins
exactly $o$ bytes into the occurrence. Unlike a certified byte here, a certified window may be vacuous, as in the window
paper. A certificate's \emph{evidence} is the occurrence it rests on, the byte itself for a certified byte and the whole
window occurrence for a certified window. The certificate is defined for every window length, while the shipped recovery
walk searches lengths two to four, a cap whose consequences Section~\ref{sec:search} takes up. The window results
inherit the scope of the cited version of the window paper, flat token sets with non-nullable patterns, which its later
versions extend to nullable sets through a positive-width equivalent; the byte results carry no such restriction.
Everything in this paper is relative to one fixed automaton, and mode-driven scanning returns in
Section~\ref{sec:repair}. Both certificates are static, computed from the automaton alone before any input exists, which
lets a recovery built on them give guarantees about inputs it has never scanned.

\section{Repair-invariant resynchronization}
\label{sec:repair}

Let $\Sigma$ be the byte alphabet throughout.
\begin{definition}[Local repair]
For input $x$ and anchor $c$ with $0 \le c \le |x|$, a \emph{repair before $c$} is any string $r \in \Sigma^{*}$, giving
the repaired input $y = r \cdot \suffix{x}{c}$. A repair is \emph{tokenizable} when $y$ is completely tokenizable.
\end{definition}

Nothing at or after the anchor changes; everything before it may, by replacement as well as insertion or deletion, so
the results below quantify over every prefix-local fix at once. The anchor is only the left edge of the text a repair
keeps, chosen by the caller or the driver, and it need not be certified, unlike a certified anchor elsewhere in this
series, which is a position some certificate makes a boundary.

Complete tokenizability to the end of input is the natural quantification for a whole file in hand, but it is stronger
than the certificates need, and the stronger form buys vacuity, since a tail whose far end no repair can save would
escape every claim below even where the near end is perfectly ordinary. The scan supplies the weaker, sharper currency.
For any input $w$, write $\con(w)$ for the scan's \emph{committed length}: the sum of the emitted token lengths when
maximal munch halts, equal to $|w|$ exactly when $w$ is completely tokenizable.

\begin{lemma}[Committed prefix]
\label{lem:committed}
For any input $w$, the committed prefix $z = \prefix{w}{\con(w)}$ is completely tokenizable, and its scan emits exactly
the tokens the scan of $w$ commits.
\end{lemma}

\begin{proof}
By induction over the committed tokens. Each commit is the last accepting position on the viable run from its token's
start, and the scan of $w$ read past it only through nonaccepting states before rolling back, so truncating $w$ at the
final commit removes no accepting candidate. At every committed start, the scan of $z$ sees the same last accept and
commits the same token, and the final commit lands exactly at $z$'s end.
\end{proof}

The lemma is the bridge the certificates cross, because whatever a scan commits is itself a completely tokenizable input
with the identical segmentation, so any theorem quantifying over completely tokenizable inputs applies to the committed
prefix of every scan, finished or failed.

\begin{definition}[Evidence-reaching repair]
\label{def:reaching}
For evidence occupying $\window{x}{q}{q+w}$ with $c \le q$, $0 < w$, and $q + w \le |x|$, a repair $r$ before $c$ is
\emph{evidence-reaching} when the scan of $y = r \cdot \suffix{x}{c}$ commits through the evidence's image:
$\con(y) \ge |r| + (q + w - c)$.
\end{definition}

Every tokenizable repair is evidence-reaching for every choice of evidence, since $\con(y) = |y|$; the reverse fails,
and that gap is the point. A repair whose scan dies far past the evidence still counts, so the quantification below
binds real, failing inputs as well as the ideal ones that tokenize to the end.

\begin{definition}[Repair-invariant resynchronization point]
Position $p$ with $c \le p < |x|$ is a \emph{repair-invariant resynchronization point} for $x$ anchored at $c$, relative
to evidence occupying $\window{x}{q}{q+w}$, if in every evidence-reaching repair $y = r \cdot \suffix{x}{c}$, a token of
the scan's committed segmentation begins at $|r| + (p - c)$.
\end{definition}

\begin{theorem}[Certified bytes resynchronize]
If $x[p]$ is a certified byte and $c \le p$, then $p$ is a repair-invariant resynchronization point for $x$ anchored at
$c$, relative to the evidence $\window{x}{p}{p+1}$, which means that in every repair whose scan commits through the
byte's image, a committed token begins at $|r| + (p - c)$.
\end{theorem}

\begin{proof}
Every repair anchored at or before $p$ preserves the occurrence of the certified byte at $p$'s image. When the scan of
$y$ commits through that image, the committed prefix $z$ contains it, $z$ is completely tokenizable with the scan's own
segmentation by Lemma~\ref{lem:committed}, and the byte certificate places a token origin at every occurrence in every
completely tokenizable input, the characterization theorem of the split-points paper~\citep[Theorem~2, arXiv
v2]{splitpoints2026}, so a token of $z$'s segmentation, which is the committed segmentation of $y$, begins at the image.
\end{proof}

\begin{theorem}[Certified windows resynchronize at their origin]
Let the token set be flat with non-nullable patterns, the scope of the cited version of the window paper. If $x$ carries
an occurrence of a certified window $(W, o)$ at position $q$, then $p = q + o$ is a repair-invariant resynchronization
point for $x$ anchored at any $c \le q$, relative to the evidence $\window{x}{q}{q+|W|}$, which means that in every
repair whose scan commits through the occurrence's image, a committed token begins at $|r| + (p - c)$.
\end{theorem}

\begin{proof}
Every repair anchored at or before $q$ preserves the occurrence whole at $q$'s image. When the scan of $y$ commits
through the occurrence's image, the committed prefix $z$ contains it whole and is completely tokenizable with the scan's
own segmentation by Lemma~\ref{lem:committed}, and the window certificate places the origin of the token covering the
occurrence's final byte at its image plus $o$ in every completely tokenizable input carrying the occurrence, the
certified-window contract and its soundness theorem~\citep[Definition~1 and Theorem~1, arXiv v1]{splitwindows2026}.
\end{proof}

\begin{corollary}[Complete repairs]
\label{cor:complete}
In every tokenizable repair $y = r \cdot \suffix{x}{c}$, a token of $y$'s segmentation begins at $|r| + (p - c)$, for
$p$ and $c$ as in either theorem.
\end{corollary}

\begin{proof}
A tokenizable repair commits its whole input, $\con(y) = |y|$, so it is evidence-reaching for any evidence, and its
committed segmentation is its segmentation.
\end{proof}

The corollary is the weaker claim. For an anchor $c$ and evidence $\window{x}{q}{q+w}$, write
\begin{gather*}
\Rcomp(x, c) = \{\, r \in \Sigma^* : \con(r \cdot \suffix{x}{c}) = |r| + |x| - c \,\}, \\
\Rreach(x, c; q, w) = \{\, r \in \Sigma^* : \con(r \cdot \suffix{x}{c}) \ge |r| + q + w - c \,\},
\end{gather*}
where the parameters are fixed per instance and dropped hereafter and $w$ is the evidence width throughout. $\Rreach$ is
the set of evidence-reaching repairs at an anchor and $\Rcomp$ the set of completely tokenizable ones. The inclusion
$\Rcomp \subseteq \Rreach$ holds always, the corollary binds only $\Rcomp$, and the theorems bind all of $\Rreach$. The
inclusion can be strict, and the example is small. Over $\{\code{ab}, \code{;}\}$ on the damaged text \code{ab;\#},
anchored at zero on the certified semicolon's evidence $[2,3)$, no repair completes the tail, the byte \code{\#} being
tokenizable in no context, so $\Rcomp$ is empty and the corollary holds vacuously. Yet the identity repair commits
\code{ab} and \code{;} through the evidence, so $\Rreach$ is inhabited and the theorem guarantees a committed interior
boundary at the image of position two in every reaching repair.

The anchor constraint is $c \le q$, not $c \le p$, since a repair may not touch the window's own bytes, which lie partly
before the resynchronization point whenever the origin is positive; at origin zero the two constraints coincide. The
searching algorithm satisfies the constraint by construction, since it only reports occurrences lying wholly in the
unmodified suffix.

\Needspace*{11\baselineskip}
\begin{lemma}[Progress]
A driver that alternates scanning with recovery, stopping at the input's end or when the recovery search refuses,
terminates after at most $|x|$ advancing iterations.
\end{lemma}

\begin{proof}
Every emitted token has positive length, and the recovery search begins one past the failure offset, so every answer
lies strictly past the position the scan held; each iteration that does not stop strictly increases the driver's offset,
which is bounded by $|x|$. Refusal and the input's end stop the driver by the loop's own condition.
\end{proof}

One-past is not pedantry, because a certified byte can sit at the failure offset. Over $\{\code{ab}, \code{;}\}$ the
byte \code{a} is certified, occurring only token-initially. On the input \code{";a;ab"} the scanner commits the first
semicolon, reads \code{a}, dies on the following semicolon with no accepting position, and halts with failure offset
$1$, exactly the certified occurrence; the failed transition itself happens one byte later. A search that began at the
failure offset would answer that same offset forever; beginning one past it, the search answers the second semicolon and
the remainder tokenizes.

\begin{remark}[Mode scope]
Under mode-driven scanning, an answer is a certificate relative to the active mode's automaton only, since the
certificates quantify over one automaton's completely tokenizable inputs, and a repair, or the resumed text itself
between the failure and the answer, may reach the resume position in a different mode, whose automaton certifies
different positions. The library consults the active mode and documents the answer as mode-relative; nothing here
upgrades the flat theorems to a guarantee over modal repairs.
\end{remark}

One more consequence connects invariance to the vacuity question. Call $x$ \emph{repairable at $c$} when some
tokenizable repair anchored at $c$ exists.

\begin{corollary}[Recovery succeeds where repair exists]
If $x$ is repairable at $c$ and $p$ is a repair-invariant resynchronization point for $x$ anchored at $c$ relative to
some evidence $\window{x}{q}{q+w}$ with $c \le q$, then $\suffix{x}{p}$ is completely tokenizable.
\end{corollary}

\begin{proof}
Take any tokenizable repair $y = r \cdot \suffix{x}{c}$; it commits its whole input, so it is evidence-reaching, a token
of $y$ begins at $|r| + (p - c)$, and the tokens of $y$ from that boundary on are a complete tokenization of $y$'s
suffix there, which is $\suffix{x}{p}$ itself.
\end{proof}

So where a complete repair exists, a repair-invariant answer begins a resumed scan that tokenizes everything that
remains. Nonvacuity alone does not buy this, because an evidence-reaching repair can commit through the evidence and
still fail later, so the complete-tokenizability conclusion needs the complete-repair premise, as stated.

The converse direction closes the loop the evaluation needs, because it gives a necessary condition for the strengthened
quantifier to outrun the corollary's at a given anchor, namely that the answer's own resumed suffix must fail to
tokenize completely.

\begin{proposition}[Reaching collapses to completing at a surviving answer]
\label{prop:collapse}
Let $p$ be a repair-invariant resynchronization point for $y$ anchored at $c$, relative to evidence $\window{y}{q}{q+w}$
with $c \le q$, and suppose $\suffix{y}{p}$ is completely tokenizable. Then every evidence-reaching repair anchored at
$c$ is a tokenizable repair, and since every tokenizable repair is evidence-reaching, the two quantifiers coincide at
$c$.
\end{proposition}

\begin{proof}
Both inclusions. Take an evidence-reaching repair $r$. By the invariance premise, a token of the committed segmentation
of $r \cdot \suffix{y}{c}$ begins at $|r| + (p - c)$, and past a committed boundary the scan restarts in the initial
state and reads exactly the image of $\suffix{y}{p}$, which is completely tokenizable by hypothesis, so the repaired
scan commits everything and $r$ is tokenizable. Conversely a tokenizable repair commits its whole input, so it commits
through the evidence's image and is evidence-reaching, the observation already made after Definition~\ref{def:reaching}.
\end{proof}

The proposition is a fence against over-reading the strengthening. At an anchor whose answer's resumed scan completes,
the new semantics binds exactly the repairs the old corollary already bound, and any vacuity of the complete-repair form
at that anchor is vacuity of the evidence-reaching form too. The proposition's premises are the walk's own
postcondition, that a certified answer is repair-invariant relative to its reported evidence, and that evidence lies at
or after the search anchor by construction. The proposition therefore applies to every certified answer whose resumed
suffix tokenizes completely in one attempt. The strengthening's strictly additional content can live only where resumed
scans fail downstream, and Section~\ref{sec:eval} reports how many of the campaign's answers remain candidates for it,
while one-attempt complete resumption decides the rest.

\begin{corollary}[Inhabited complete repairs force the collapse]
\label{cor:dichotomy}
Let $p$ be repair-invariant for $x$ anchored at $c$ relative to evidence $\window{x}{q}{q+w}$ with $c \le q$. If
$\Rcomp(x, c) \ne \varnothing$, then $\suffix{x}{p}$ is completely tokenizable and $\Rcomp = \Rreach$. Equivalently,
\[
\Rcomp \subsetneq \Rreach \quad\text{if and only if}\quad \Rcomp = \varnothing \text{ and } \Rreach \ne \varnothing.
\]
\end{corollary}

\begin{proof}
Take a complete repair. It is evidence-reaching, so the invariance premise on $p$ places a committed boundary at $p$'s
image; the repair tokenizes its whole input, so the suffix from that boundary, which is $\suffix{x}{p}$, tokenizes
completely, and Proposition~\ref{prop:collapse} turns every evidence-reaching repair into a complete one. To prove the
equivalence, suppose the inclusion is strict. The two sets differ, so the first part forces $\Rcomp = \varnothing$, and
$\Rreach$ properly contains it and is inhabited. Conversely, an empty $\Rcomp$ inside an inhabited $\Rreach$ is strict
by definition.
\end{proof}

So the strengthening can add content in exactly one regime. Whether a given candidate answer lies in that regime turns
on $\Rreach$'s inhabitation, which Section~\ref{sec:exact} decides from the preserved tail and which the campaign's
archive, carrying no tails, does not.

\subsection{Non-claims}

Four fences, each deliberate. A repair that modifies the suffix at or after the anchor falls outside the repair relation
proved here, since the definition quantifies over prefix repairs only and, for the window theorem, the anchor must
additionally sit at or before the occurrence. The theorems say nothing about such an edit either way, though one that
happens to preserve the evidence still faces the occurrence-universal certificate. The certificate alone does not decide
whether an evidence-reaching repair exists, and Section~\ref{sec:exact} decides it from the preserved tail. Where no
scan of any repair commits through the evidence, the guarantee holds vacuously. That boundary is strictly tighter than
one drawn at complete tokenizability, and the residual asymmetry between the certificate kinds narrows with it. A byte
answer's evidence is the answer itself, so whenever the damaged suffix admits a first commit at all, that commit
witnesses content for the anchor placed at the answer. A window with positive origin instead rests on evidence beginning
before the resume position, which no scan anchored at the answer reads, so its guarantee can bind repairs the resumed
scan never exhibits. The answer is the first certificate in evidence order, not the smallest answerable position, since
a window met earlier can answer a byte or two past one met later. Finally, on a suffix no repair can save to the end of
input, the theorems still bind every repair whose scan commits through the evidence, but the driver guarantees
termination and nothing else, each resume strictly advancing until the input ends or the search refuses.

\section{Deciding the quantifier}
\label{sec:exact}

The theorems above bind every evidence-reaching repair and say nothing about whether one exists. This section decides
that question, and with it invariance at every endpoint, from the preserved tail alone, at the cost of a finite family
of repairs read off the automaton, at most one more than its state count. The family also yields a shortest reaching
repair, a shortest repair refuting any proposed position, the single threshold that governs nonvacuous invariance, and
the exact relation to complete repairs. The shipped repair routines turn out to be the family restricted to complete
repairs, which proves their documented contracts and with them the negative labels the evaluation consumes.

\subsection{The tail, the repairs, and what a repair shows}

Translate the anchor to zero. The preserved tail is $t = \suffix{x}{c}$, of length $m > 0$, and a repair $r$ gives the
input $rt$, whose \emph{seam} is the offset $|r|$ where the tail begins. The automaton has $n \ge 1$ states, initial
state $q_0$, a partial transition function $\delta$ over $\Sigma$, and accepting states carrying their token kinds. The
count $n$ is of the automaton as given, an accepting initial state allowed, at which the scan emits no empty token. For
a repair $r$ write
\[
e(r) = \con(rt) - |r|, \qquad B(r) = \{\, p \in [0, m) : \text{a committed token of } rt \text{ begins at } |r| + p
\,\},
\]
the relative terminal and the in-tail committed starts. The start of a failed token attempt is not a committed start,
and neither is the end of input, so $e(r) \notin B(r)$. Evidence $\window{x}{q}{q+w}$ at anchor $c$ has the endpoint
$h = q + w - c$, with $1 \le h \le m$, and the two sets of Section~\ref{sec:repair} read
\[
\Rreach = \{\, r : e(r) \ge h \,\}, \qquad \Rcomp = \{\, r : e(r) = m \,\}.
\]
Only the endpoint enters, since where the evidence begins does not change which repairs reach it. A position
$p \in [0, m)$ is \emph{invariant at $h$} when $p \in B(r)$ for every $r \in \Rreach$, vacuously when $\Rreach$ is
empty. The definition admits $p \ge h$ as well as positions inside the evidence.

\subsection{Finitely many repairs decide everything}

Call a state \emph{positively reachable} when some nonempty word drives $q_0$ to it. For each such state $q$ fix a
shortest nonempty word $u_q$, which a breadth-first search seeded with $q_0$'s outgoing transitions, $q_0$ itself
unvisited, finds for every state at once. The initial state is a candidate like any other. When a nonempty word returns
to it, $u_{q_0}$ is a shortest positive return, and the empty word does not stand in for it. Continue the run from $q$
over the whole of $t$, stopping where a transition is missing, and let $d_q \in [0, m]$ be its last accepting offset,
offset zero counting when $q$ itself accepts. A state whose run never accepts contributes nothing.

\begin{definition}[Representative repairs]
\label{def:representatives}
The representatives of $t$ are the empty repair and, for every positively reachable $q$ whose run over $t$ accepts, the
repair $u_q$, whose first committed token ends at tail offset $d_q$. There are at most $n + 1$ of them.
\end{definition}

The name is justified by what each representative does. The first token of $u_q t$ ends at $|u_q| + d_q$, because the
run's acceptance at $d_q$ comes after every acceptance inside $u_q$ and by definition none comes later; the scan then
restarts in $q_0$ at tail offset $d_q$. When $d_q = 0$, the state $q$ accepts and nothing later on its run does, so the
first token ends exactly at the seam and the rest is the fresh scan of $t$, the empty repair's own behavior. Such a
representative is an actual repair that adds nothing, kept for uniformity, and may be dropped. For any repair whose scan
restarts at tail offset $d$, the rest is the fresh scan of $\suffix{t}{d}$, so $e$ and $B$ are read off that scan.

\begin{theorem}[Finite exactness]
\label{thm:finite}
Every repair $r$ with $e(r) \ge 1$ has the same relative terminal $e$ and the same in-tail committed starts $B$ as some
representative of length at most $|r|$. Hence for every endpoint $h$ the representatives decide whether $\Rreach$ is
inhabited and which positions are invariant. They also contain a shortest reaching repair and a shortest repair refuting
any position that is not invariant.
\end{theorem}

\begin{proof}
The scan of $rt$ commits through the first tail byte, so one committed token covers it. Let it begin at $s$ and write
$a = |r|$ for the seam. If $s = a$, the scan stands in $q_0$ at the seam and everything it commits from there is the
fresh scan of $t$, so the empty repair has the same $e$ and $B$ and length zero. Otherwise $s < a$. Put $u = r[s..a)$,
nonempty, and $q = \delta^*(q_0, u)$, defined because $u$ is a prefix of a committed token. The token ends at $a + d$
with $d > 0$, since it covers the first tail byte. By longest match, $d$ is the last accepting offset of the run from
$q$ over $t$, because a later acceptance before the run died would have ended the token there. Hence $q$ is positively
reachable, its run accepts, and $d_q = d$. Now scan $u_q t$. The run from $q$ over $t$ is the same, its acceptance at
$d > 0$ comes after every acceptance inside $u_q$, and the first token ends at $|u_q| + d$ with the same kind. Both
scans then restart in $q_0$ on $\suffix{t}{d}$ and commit the same tokens to the same terminal, while the tokens of $rt$
before the crossing token all begin before the seam and are not in $B(r)$. Hence $e(u_q) = e(r)$ and $B(u_q) = B(r)$,
with $|u_q| \le |u| \le |r|$. The decisions follow. A reaching repair exists exactly when a representative reaches, and
a position is invariant exactly when every reaching representative commits a token there. A shortest witness of either
kind can be taken among the representatives, since every witness maps to one that is no longer and every representative
is itself a repair.
\end{proof}

The theorem preserves the kinds and lengths of the tokens that begin in the tail, but not the crossing token's length or
where it began, which change with the repaired prefix. Leaving out a state whose run over $t$ never accepts does not say
that $u_q t$ fails to reach, since its scan may commit inside $u_q$ and cross the seam in another state, which then
represents it.

\begin{theorem}[Sharp witness bounds]
\label{thm:bounds}
If $\Rreach$ is inhabited, it contains a repair of length at most $n - 1$. If a position is not invariant, some reaching
repair of length at most $n$ refutes it, and one of length at most $n - 1$ when the position is not $0$. All three
bounds are attained for every $n \ge 1$, already over two symbols.
\end{theorem}

\begin{proof}
A shortest nonempty word to a state other than $q_0$ repeats no state, so $|u_q| \le n - 1$. A shortest positive return
to $q_0$ is a simple cycle, so $|u_{q_0}| \le n$. With Theorem~\ref{thm:finite} this gives the counterrepair bound. For
reaching alone, the return to $q_0$ is not needed. If $d_{q_0} > 0$, then $d_{q_0}$ is the last accepting offset of the
run from $q_0$ over $t$, which is where the empty repair's first token ends, so the empty repair restarts at the same
offset and has the same terminal. If $d_{q_0} = 0$, the representative's scan is the fresh scan of $t$, as noted above.
An accepting $q_0$ changes nothing here, since its acceptance at offset zero is where the token of $u_{q_0}$ ends, while
the empty repair emits no empty token there, and the terminals still agree. A shortest reaching repair is the empty one
or some $u_q$ with $q \ne q_0$, of length at most $n - 1$. The same comparison serves a position $p > 0$, where by
Theorem~\ref{thm:finite} a shortest refuter is again a representative and only $u_{q_0}$ can have length $n$. The empty
repair and $u_{q_0}$ restart at the same offset and commit the same starts from there on, and the only start they can
differ in is $0$, where the empty repair's first token begins when $d_{q_0} > 0$. A reaching $u_{q_0}$ that refutes
$p > 0$ is matched by the empty repair, which reaches and refutes $p$ as well, so a shortest refuter of such a position
has length at most $n - 1$.

Tightness. For $n = 1$, the reaching bound is zero and the empty repair attains it wherever anything reaches. For the
counterrepair bound at $n = 1$, take one accepting state with a loop on \code{a}. It recognizes $\code{a}^{*}$, and
since the scan emits no empty token its tokens are $\code{a}^{+}$. Take the tail \code{a} with $p = 0$. The empty repair
emits the tail as a token beginning at zero, while the repair \code{a} makes \code{aa} one token, so no token begins at
the seam, and no shorter nonempty repair exists. For $n \ge 2$, take the chain
$q_0 \xrightarrow{a} q_1 \xrightarrow{a} \cdots \xrightarrow{a} q_{n-1}$ with a loop on \code{b} at $q_{n-1}$, which
alone accepts, and the tail \code{b}. The repair $\code{a}^{n-1}$ reaches, and any token covering the tail's \code{b}
reads it from $q_{n-1}$, the only state with a \code{b} transition, so the part of the repair inside that token has at
least $n - 1$ symbols. For the counterrepair bound, take the cycle $q_i \xrightarrow{a} q_{(i+1) \bmod n}$, the one
extra transition $q_0 \xrightarrow{b} q_1$, and $q_1$ alone accepting, with the tail \code{b} and $p = 0$. The empty
repair emits \code{b} at zero. A refuting repair needs a crossing token that reads the tail's \code{b}, which only $q_0$
can, so the part of the repair inside that token is a positive return to $q_0$, of length at least $n$; the repair
$\code{a}^{n}$ makes $\code{a}^{n}\code{b}$ one token and attains it. For the bound at a position $p > 0$, take the
chain family, at $n = 1$ the single accepting state with a loop on \code{b}, with the tail \code{bb} and $p = 1$. At
$n = 1$ the empty repair emits \code{bb} as one token and refutes $1$ at length zero. For $n \ge 2$, as for the tail
\code{b}, a reaching repair has a crossing token that reads the tail from $q_{n-1}$, so the part of the repair inside it
has at least $n - 1$ symbols, and that token absorbs both \code{b}s, so $\code{a}^{n-1}$ refutes $1$ and no shorter
repair does. Making $q_0$ accepting in either family changes nothing, since no empty token is emitted, so the bounds are
attained with an accepting initial state as well, and $n$ counts the automaton as given.
\end{proof}

\subsection{One threshold}

For $i < m$ let $J(i) > i$ be the end of the first token a fresh scan of $\suffix{t}{i}$ commits, undefined when that
attempt fails, and let $m$ be a node too. The edges $i \to J(i)$ point toward larger offsets and form a forest whose
roots are the failed attempts and the end of input; let $\rho(i)$ be the root reached from $i$. A repair restarting at
$d$ has $e = \rho(d)$, and its $B$ is the path from $d$ to $\rho(d)$ with the root removed.

\begin{lemma}[Common starts]
\label{lem:common}
If two repairs share a committed start $p$ in the tail, they have the same terminal and the same committed starts from
$p$ on.
\end{lemma}

\begin{proof}
At a committed start the scanner is in $q_0$ and reads $\suffix{t}{p}$, the same fixed suffix for both, so it commits
the same tokens to the same terminal. In the forest, both continue along the one path from $p$ to its root.
\end{proof}

Hence repairs with different terminals have disjoint start sets. Let $E$ be the largest terminal of a representative,
and let $F$ be the largest terminal below $E$, or $0$ when no positive one exists. The representatives with terminal $E$
lie in one tree; let $g$ be the lowest common ancestor of their entries, and let $C$ be the path from $g$ to $E$ with
$E$ removed, empty when $g = E$.

\begin{theorem}[The evidence profile]
\label{thm:profile}
Let $1 \le h \le m$. If $h > E$, no repair reaches and every position is vacuously invariant. If $h \le F$, repairs
reach and no position is invariant. If $F < h \le E$, repairs reach and the invariant positions are exactly $C$.
\end{theorem}

\begin{proof}
The first case is Theorem~\ref{thm:finite} and the maximality of $E$. In the second, a representative with terminal $E$
and one with terminal $F$ both reach $h$, and their start sets are disjoint by Lemma~\ref{lem:common}. In the third, the
reaching representatives are exactly those with terminal $E$, since no terminal lies strictly between $F$ and $E$. By
Theorem~\ref{thm:finite} the invariant positions are the common starts of these representatives, the intersection of
their paths to $E$, which is the path from their lowest common ancestor $g$ to $E$; removing $E$ removes the one node on
it that is not a committed start.
\end{proof}

Among the endpoints some repair reaches, the invariant set changes at most once, at $F + 1$. The reaching family may
change several times below $F$, but its common starts stay empty there, and every position that is ever nonvacuously
invariant becomes so at $F + 1$, while positions outside $C$ never do. If the evidence must also contain a position $p$
of $C$, as a byte's or a window's evidence does, its least nonvacuous endpoint is $\max(F + 1, p + 1)$.

\begin{corollary}[Reaching against completing]
\label{cor:profile-complete}
If $E = m$, then $\Rreach = \Rcomp$ exactly when $h > F$. If $E < m$, then $\Rcomp = \varnothing$, and
$\Rreach = \Rcomp$ exactly when $h > E$. Proposition~\ref{prop:collapse} and Corollary~\ref{cor:dichotomy} follow.
\end{corollary}

\begin{proof}
A repair completes exactly when its terminal is $m$. If $E = m$, some representative completes, and $\Rreach$ exceeds
$\Rcomp$ exactly when some reaching repair stops short of $m$. By Theorem~\ref{thm:finite} that is exactly when a
representative has terminal in $[h, m)$, that is, $h \le F$. If $E < m$, no repair completes, so the two sets are equal
exactly when $\Rreach$ is empty, which is $h > E$. For the proposition, let $p$ be invariant and let $\suffix{t}{p}$
tokenize completely. If $h > E$, both sets are empty. Otherwise invariance is not vacuous, so $F < h \le E$ and
$p \in C$, whence $\rho(p) = E$. The complete fresh scan from $p$ gives $\rho(p) = m$, so $E = m$, and $h > F$ gives
$\Rreach = \Rcomp$. For the dichotomy, let $p$ be invariant. If $\Rcomp$ is inhabited, a complete repair reaches, so
$E = m \ge h$ and invariance is not vacuous. Then $F < h \le E$ and $p \in C$, whence $\rho(p) = m$, so the suffix
$\suffix{t}{p}$ tokenizes completely, and $\Rreach = \Rcomp$ by the first part. For the equivalence, if $h > E$, both
sets are empty, so they are equal and $\Rreach$ is not inhabited. If $h \le E$, then $F < h$ as before, and the
inclusion is strict exactly when $E < m$, which is exactly $\Rcomp = \varnothing$ with $\Rreach \ne \varnothing$.
\end{proof}

\begin{table}[tbp]
\centering
\small
\renewcommand{\arraystretch}{1.15}
\makeatletter
\begin{tabular}{llll}
\toprule
Repair & Committed tokens of $rt$ & $B(r)$ & $e(r)$ \\
\midrule
$\varepsilon$ & \code{bbb}\,\code{a} & $\{0,3\}$ & $4$ \\
\code{bb} & \code{bbb}\,\code{bb}\,\code{a} & $\{1,3\}$ & $4$ \\
\code{b} & \code{bbb} & $\varnothing$ & $2$ \\

\bottomrule
\end{tabular}
\par\medskip
\begin{tabular}{lll}
\toprule
Endpoint & Reaching behaviors & Invariant positions \\
\midrule
$1 \le h \le 2$ & $\varepsilon$, \code{bb}, \code{b} & $\varnothing$ \\
$3 \le h \le 4$ & $\varepsilon$, \code{bb} & $\{3\}$ \\

\bottomrule
\end{tabular}
\makeatother
\renewcommand{\arraystretch}{1.0}
\caption{The basis example, $\{\code{a}, \code{bb}, \code{bbb}\}$ on the tail \code{bbba}. Above, the three reaching
behaviors, one actual repair each, with the committed tokens of the repaired input, the in-tail starts, and the
terminal. The empty repair is $\varepsilon$. Below, which behaviors reach each endpoint and the exact invariant set
there. Two behaviors are needed at $h = 3$ and all three across the endpoints. Both bodies are written by the program of
Section~\ref{sec:exact} and held to it by the test suite.}
\label{tab:basis}
\end{table}

\begin{table}[tbp]
\centering
\small
\renewcommand{\arraystretch}{1.15}
\makeatletter
\begin{tabular}{lllllll}
\toprule
Tokens & Tail & $m$ & $E$ & $F$ & $C$ & Basis \\
\midrule
$\{\code{ab}, \code{;}\}$ & \code{ab;\#} & $4$ & $3$ & $0$ & $\{0,2\}$ & $\varepsilon$ \\
$\{\code{ab}, \code{ba}\}$ & \code{ab} & $2$ & $2$ & $1$ & $\{0\}$ & $\varepsilon$, \code{b} \\
$\{\code{a}, \code{ba}\}$ & \code{a} & $1$ & $1$ & $0$ & $\varnothing$ & \code{b} \\
$\{\code{a}, \code{bb}, \code{bbb}\}$ & \code{bbba} & $4$ & $4$ & $2$ & $\{3\}$ & $\varepsilon$, \code{bb}, \code{b} \\

\bottomrule
\end{tabular}
\makeatother
\renewcommand{\arraystretch}{1.0}
\caption{Evidence profiles of the examples in this paper: the tail length, the largest terminal $E$, the next one $F$,
the invariant positions $C$ at every endpoint in $(F, E]$, and the repair basis of Theorem~\ref{thm:basis}. The first
row is the strict-inclusion example of Section~\ref{sec:repair}, the second the separation example of
Section~\ref{sec:eval}, the third the equal-terminals example, and the last the basis example. Written by the program of
Section~\ref{sec:exact}.}
\label{tab:profiles}
\end{table}

Equal terminals alone force no agreement. Over $\{\code{a}, \code{ba}\}$ on the tail \code{a}, both the empty repair and
\code{b} complete, the first emitting \code{a} at zero and the second absorbing the tail into \code{ba}. Here $E = m$,
$F = 0$, and $C$ is empty, as Table~\ref{tab:profiles} records. The same table profiles two other examples from the
paper. On \code{ab;\#} over $\{\code{ab}, \code{;}\}$ the only reaching behavior is the empty repair's, with
$E = 3 < m$, so positions $0$ and $2$ are invariant at every endpoint up to $3$, the strict inclusion of
Section~\ref{sec:repair} at all three. On \code{ab} over $\{\code{ab}, \code{ba}\}$ the repair \code{b} ends at $F = 1$,
so the one-byte evidence has no invariant position while the whole tail has position $0$, the separation
Section~\ref{sec:eval} draws between the decider and the certificates.

\subsection{Two repairs, then three}

\begin{theorem}[A finite consensus basis]
\label{thm:basis}
At every endpoint with $\Rreach$ inhabited, one or two reaching repairs have exactly the common starts of all of
$\Rreach$. One fixed set of at most three repairs reproduces inhabitation and the invariant set at every endpoint at
once. Every repair chosen has length at most $n$, and both counts are tight.
\end{theorem}

\begin{proof}
If $h \le F$, a representative with terminal $E$ and one with terminal $F$ have disjoint starts, settling every
position. If $F < h \le E$ and some representative enters at $g$, its path alone is $C$. Otherwise every entry with
terminal $E$ lies strictly below $g$, and not all in one child subtree of $g$, since a child containing them all would
be a common ancestor below $g$. Two entries from two different children have lowest common ancestor $g$, so their paths
meet in exactly $C$. For the uniform set, keep those one or two representatives and add one with terminal $F$ when
$F > 0$. Above $E$ nothing reaches, between $F$ and $E$ only the kept maximal ones reach and give $C$, and at or below
$F$ the added one empties the intersection while at least one repair reaches. Tightness is Table~\ref{tab:basis}, where
over $\{\code{a}, \code{bb}, \code{bbb}\}$ on the tail \code{bbba} there are exactly three reaching behaviors. A token
crossing the seam reads the tail's first \code{b} at a noninitial position, so by longest match it is \code{bbb} with
one or two of its \code{b}s before the seam. It ends at tail offset $2$, where the scan of \code{ba} fails, or at offset
$1$, after which the scan emits \code{bb} and \code{a}. A token ending or beginning at the seam gives the empty repair's
behavior. At $h = 3$ the invariant set is $\{3\}$, and neither terminal-$4$ behavior has it alone, so two repairs are
needed at that endpoint. At $h = 1$ those two still meet in $\{3\}$ while the third empties the intersection, so no pair
serves every endpoint.
\end{proof}

\subsection{Cost}

\begin{theorem}[One anchor]
\label{thm:cost}
For one tail, the evidence profile $(E, F, g, C)$ with its basis of at most three repairs is computed in
$O(nm + n|\Sigma|)$ time and $O(n + m)$ words beyond the automaton and the tail. Afterwards, inhabitation and invariance
at any endpoint and position are $O(1)$. A shortest reaching repair, or a shortest counterrepair to a given position, is
selected in $O(n)$ time and written out in time proportional to its length.
\end{theorem}

\begin{proof}
For $0 \le i \le m$ and state $q$ let $L_i(q)$ be the last accepting offset of the run from $q$ over $\suffix{t}{i}$,
offset $i$ allowed, undefined when the run never accepts. Then $L_m(q) = m$ for accepting $q$ and is undefined
otherwise. For $i < m$, $L_i(q)$ is $L_{i+1}(\delta(q, t[i]))$ when that transition and value exist, else $i$ when $q$
accepts, else undefined. A later acceptance exists exactly when the value after the next byte does, and then it
dominates. Two rows of $n$ entries compute every row right to left in $nm$ steps. Keep $J(i) = L_i(q_0)$ where it
exceeds $i$, which excludes an empty token at an accepting $q_0$, and keep the row $L_0$, whose entries are the $d_q$.
Then $\rho$ follows right to left in $O(m)$. The breadth-first search stores one parent transition per state,
$O(n|\Sigma|)$ time and $O(n)$ words, and no representative is written out until asked for. For $g$, put one unit at
each entry with terminal $E$, duplicates counted, and push the counts along $J$ in increasing offset order. Every node
below a node is processed before it, so the first node carrying the full count is $g$, and the chain from $g$ gives $C$.
Carrying one entry index with the counts names an entry at $g$, or one from each of two children of $g$ with positive
counts. That is the basis, and when $F > 0$ an entry with terminal $F$ is added. Reversing $J$ and walking each tree
assigns every node a preorder interval, so whether $p$ lies on a representative's path is $p < \rho(d)$ together with an
interval test, $O(1)$. The witness selections scan at most $n + 1$ representatives with that test, and the parent
transitions write out the chosen word. Every array has $O(n + m)$ entries.
\end{proof}

The jump table of \code{recovery.cpp} is built differently, by running the scanner afresh from every tail offset and
reading each run to its end, which is quadratic in the tail on a long token. On the two-state automaton of
$\code{a}^{+}$ over a tail of \ExactUnaryLength{} bytes the shipped construction takes \ExactShippedSteps{} transitions
where the recurrence updates \ExactRecurrenceCells{} cells. The counts cover the tabulation alone and give no timing for
either routine.

\begin{theorem}[Every anchor]
\label{thm:anchors}
For one input of length $m$, the evidence profiles and bases of all its anchors are computed in $O(nm + n|\Sigma|)$ time
and $O(n + m)$ words. Afterwards, at every anchor, inhabitation, invariance, the earliest nonvacuous invariant start,
the least nonvacuous endpoint of a position, and equality of the two quantifiers are $O(1)$.
\end{theorem}

\begin{proof}
Every anchor's tail ends at the same end of input, so the recurrence, $J$ and $\rho$ are shared. One right-to-left pass
builds them, and the breadth-first search is shared too. At anchor $c$ the entries are $c$ for the empty repair and
$L_c(q)$ for each positively reachable $q$ where defined, so a second right-to-left pass, recomputing the rows while
retaining only two at a time, has every anchor's entries in turn. Offsets are absolute from here on, with $E_c$, $F_c$
and $g_c$ offsets into the input and $F_c$ defaulting to $c$. The relative endpoint $h$ corresponds to the absolute
endpoint $c + h$, where the evidence ends. Among the entries with terminal $E_c$ keep only the earliest and the latest
in preorder. Their lowest common ancestor is that of the whole set. A subtree is a contiguous preorder interval, so the
subtree of their common ancestor contains every entry between them, and any common ancestor of all entries is an
ancestor of the two. One constant-time lowest-common-ancestor query then gives $g_c$. The two kept entries are the
basis, with an entry of terminal $F_c$ added when $F_c > c$. The query comes from the classical static index, which
attaches the roots to one virtual root, records an Euler tour with its depths, and answers a query by the minimum depth
between two first visits. Adjacent depths differ by one, and the range-minimum index of Bender and
Farach-Colton~\citep{bender2000lca} takes $O(m)$ preprocessing and answers in $O(1)$. The second pass does $O(n)$ work
per anchor and one query. The pass keeps two rolling rows, the forest with its index, the parent transitions, and a
constant-size evidence profile per anchor. The queries read the profile: inhabitation is $c + h \le E_c$; invariance is
$c + h > E_c$, the vacuous case, or $F_c < c + h \le E_c$ with $p$ on the chain from $g_c$, which is $p < E_c$ and an
ancestor test; the earliest start is $g_c$ unless $g_c = E_c$; the least nonvacuous endpoint of a position of $C_c$ is
the absolute $F_c + 1$, or $\max(F_c + 1, p + 1)$ when the evidence must contain $p$, and a position outside $C_c$ has
none; and equality of the quantifiers is Corollary~\ref{cor:profile-complete}.
\end{proof}

The bound is for recovery decisions on a fixed input and automaton, in the word model with the transition table given,
and does not bound the scan itself; it says nothing about edits to the input, a change of automaton, or a change of
mode.

\subsection{The shipped routines are the complete case}

For complete repairs, at $h = m$, the routines the evaluation runs, \code{minimal\_repair} and
\code{next\_anchored\_start}, use the representatives of Definition~\ref{def:representatives}. Their crossing entries
are the positively reachable states with shortest words, found by the same seeded breadth-first search, and each entry's
boundary is $d_q$. Their jump table carries $J$, exact because the start these routines use is nonaccepting, together
with whether the fresh scan from each offset reaches $m$, which is $\rho(i) = m$. \code{minimal\_repair} returns the
empty repair when $\rho(0) = m$, otherwise a shortest $u_q$ with $\rho(d_q) = m$, and refuses when there is none.
\code{next\_anchored\_start} intersects the start sets of the representatives with terminal $m$ and returns the first
member at or after the given offset, refusing when no representative completes. Both are members of the simulator in
\code{libs/dfa/src/simulator.cpp} at the campaign's release v1.6.0 and free functions over it in
\code{libs/dfa/src/recovery.cpp} since release v2.0.0, the same algorithm in both.

\begin{theorem}[Negative-label exactness]
\label{thm:labels}
For every non-nullable token set, every tail $t$, and every anchor, $\code{minimal\_repair}(t) = \code{nullopt}$ implies
$\Rcomp = \varnothing$, and a returned repair is a shortest complete one. When some complete repair exists,
$\code{next\_anchored\_start}(t, s)$ returns the least position at or after $s$ at which every complete repair commits a
token, and refuses when none lies at or after $s$. It refuses when no complete repair exists. Since release v2.0.0 the
same holds for every token set.
\end{theorem}

\begin{proof}
On an empty tail \code{minimal\_repair} returns the empty repair and \code{next\_anchored\_start} refuses, both
correctly, because the empty repair completes and there is no position. Let $m > 0$ and take $h = m$, so that
$\Rreach = \Rcomp$. By Theorem~\ref{thm:finite}, a complete repair exists exactly when a representative has terminal
$m$, and a shortest complete repair is a shortest such representative. The positions every complete repair commits a
token at are the common starts of those representatives, and the routines compute exactly these. A refusal of
\code{minimal\_repair} therefore certifies that no repair of any length completes, and the decider's answer is the least
invariant position at or after $s$. The campaign's release refused every nullable token set before computing anything,
which is why the statement is restricted to non-nullable sets, every campaign automaton among them. Since release v2.0.0
the library compiles a nullable set with a fresh nonaccepting start, one state more. The compiled automaton's positively
reachable states, shortest words and $d_q$ are those of the given automaton, so the representatives, and with them the
theorem, are the same for every token set.
\end{proof}

\paragraph{Held to a program.} Every table and number of this section is written by \code{repair\_profile.py} under
\code{paper/panic-mode/checks}, a program with no dependencies that runs in the test suite, where its output is compared
byte for byte against the committed files. The program builds the representatives, the forest, the evidence profile and
the basis, and checks every decision against brute force over \ExactInstances{} automaton and tail instances,
\ExactQueries{} endpoint and position queries and \ExactScans{} scans. The brute-force check scans every repair up to a
length bound with the program's own maximal-munch scanner. The instances are random partial automata and random literal
vocabularies, the two tight families, the basis example, and the worked examples of Table~\ref{tab:profiles}. The
families are checked at every state count up to \ExactFamilyBrute{} by brute force and up to \ExactFamilyMax{} by the
construction, each with and without an accepting initial state. The basis example is checked to repair length
\ExactBasisBound{}, where the three behaviors of Table~\ref{tab:basis} are the only reaching ones. The shipped routines,
transcribed from \code{recovery.cpp}, are held to the $h = m$ case on every instance, with an accepting start unrolled
first, as the library does. The test checks that transcription, and the C++ itself is checked by
\code{munch\_crosscheck\_anchor}, the recovery cross-check of Section~\ref{sec:eval}. The bounded runs are evidence for
the program, while the proofs carry the quantifiers.

\section{The search}
\label{sec:search}

\begin{figure}[htbp]
\begin{lstlisting}[style=walk]
search(x, s):                                  $\triangleright$ the forward walk, evidence order
  for p := s to |x| - 1:
    if certified_byte(x[p]):                   $\triangleright$ a byte answers at its own position
      return p
    for len := 2 to 4, while p + len <= |x|:   $\triangleright$ the shipped window lengths
      o := certified_window(x[p .. p + len))
      if o != none:
        return p + o                           $\triangleright$ occurrence plus certified origin
  return none                                  $\triangleright$ explicit refusal

driver(x):
  pos := 0
  while pos < |x|:
    scan from pos, emitting tokens
    if the scan committed everything: stop
    f := the failed scan's committed offset
    p := search(x, f + 1)                      $\triangleright$ the search starts one past f
    if p = none: stop                          $\triangleright$ refusal, never a guess
    pos := p                                   $\triangleright$ resume at the returned position
\end{lstlisting}
\caption{The search and the canonical driver loop. The walk returns the first certificate in evidence order, the byte at
its own position before windows by increasing length at their occurrence plus origin, and refusal is explicit. The
search starts one past the failure offset, the progress lemma's premise, and resumption is at the returned position.
\code{next\_certified\_start()} implements the walk and \code{recover()} performs one recovery move, while the loop and
its stop policy belong to the caller, shown here as pseudocode.}
\label{fig:walk}
\end{figure}

The primitive is one forward walk (Figure~\ref{fig:walk}) that consults both certificate kinds at every position. A
certified byte answers at its own position, and a searched window answers at its occurrence plus the certified origin.
The walk returns the first certificate met in evidence order, which is the position order of the supporting evidence,
hence the first-in-evidence-order fence above. At one position the byte certificate is consulted first, then windows by
increasing length, two through four. One design choice departs from the companions' planners. The planners disable the
window search whenever byte certificates exist, because a planner prices whole-input cuts and byte certificates are
strictly cheaper there. A recovery wants an early sound point of either kind after a specific offset, so the walk keeps
both kinds live throughout. In the campaign's release a nullable token set contributes no windows, since the window
proof of the cited version excludes nullability, while its byte certificates, when any exist, stand and answer alone.
The library now compiles a nullable set with the fresh nonaccepting start of Section~\ref{sec:exact}, so its windows are
decided too, and the walk consults them like any other. Window membership is memoized per distinct byte string as in the
planner. When no searched certificate exists at or after the offset there is no answer, and the evaluation counts that
refusal. Refusal is relative both to the searched lengths and to the window paper's conservative model. A shipped
fixture shows the first, carrying a five-byte certificate exactly one past the cap, where the walk refuses by design.
The second arises because the model can refuse a semantically certified window. Over $\{\code{a}, \code{ab}, \code{b}\}$
every occurrence of \code{ab} begins a token at origin zero, yet the shipped decider refuses it, so a refusal asserts
only that no certificate of the searched kinds is reported, never that no sound position exists.

\begin{corollary}[Walk soundness]
If the search starting at offset $s$ returns $p$ on evidence occupying $\window{x}{q}{q+w}$ with $q \ge s$, then $p$ is
a repair-invariant resynchronization point for $x$ anchored at any $c \le q$, hence at any $c \le s$, relative to that
evidence, which means that in every repair whose scan commits through the evidence's image, a committed token begins at
the answer's image.
\end{corollary}

\begin{proof}
By cases on the evidence. A certified byte answers at its own position, $q = p$ with $w = 1$, and $c \le q = p$ is the
byte theorem's anchor condition. A certified window answers at its occurrence plus the certified origin with the
occurrence at $q$ and $w = |W|$, and $c \le q$ is the window theorem's. The walk reports only evidence lying wholly at
or after $s$, so $c \le s$ always suffices.
\end{proof}

The following worked answer is taken verbatim from the library's own pinned fixture. Over identifiers and whitespace
runs no byte usefully certifies (the at-sign, which no token contains, is formally certified only in the vacuous sense
and reported by nothing), so recovery rests on windows alone. On the input \code{"abc@@@def ghi"} the scan fails inside
the junk. In the forward search, the first certificate met is the four-byte window \code{"def\textvisiblespace"} at
offset $6$ with certified origin $3$, so the answer is $9$, the whitespace byte that begins a token in every completely
tokenizable input containing the occurrence. The window's bytes at offsets $6$ through $9$ sit wholly in the unmodified
suffix, which is the anchor constraint of the window theorem, so a repair of everything before offset $6$ is bound by
the guarantee, and a repair that rewrote \code{def} would not be. A recovery that returned the occurrence instead of
occurrence plus origin would resume at $6$, a position the window does not certify and one that can lie mid-identifier
in a tokenizable repair; the shipped fixture kills that mutant.

\section{Implementation}
\label{sec:impl}

The munch library~\citep{munch} ships the procedure in two layers, and the split is load-bearing. The cited release
states the contract in its API documentation and in the repository's top-level readme. The primitive,
\code{Lexer::next\_certified\_start(input, from)}, is a position-only query on one automaton, documented with the
complete-repair corollary's contract, and it knows nothing about errors or drivers. Beside it, the evidence-returning
\code{next\_certified\_evidence(input, from)} carries the interval and kind behind each answer. The behavior,
\code{Tokenizer::recover()}, belongs to the driver. A failed \code{next()} does not advance the reading position, since
guessing a skip would invent tokens, and the driver then chooses among stopping, seeking by its own rule, or
\code{recover()}, which asks the active mode's lexer for the first certified start at or after the position one past the
current one and moves there, returning the skip count. When no certificate among the searched kinds and lengths lies
ahead, the position does not move and the refusal is explicit. Two sibling forms return the evidence itself:
\code{recover\_from\_failure()} answers with the certified start and its evidence interval, and
\code{recover\_from\_clean(clean\_from)} floors the search at a caller's known-clean offset, the returned evidence
covered by construction, the interface the trust discussion of Section~\ref{sec:eval} measures.

The fail-closed tests pin each half at its own layer. At the lexer, fixtures pin both certificate kinds answered in
evidence order, the nullable set answering through its certified byte with no window consulted, and the unbounded run
refusing. At the driver, a mutant that ignored the window origin would move to the uncertified occurrence, and a named
fixture kills it. The driver's refusal path and its exact skip counts are pinned, and a driver fixture holds recovery to
the active mode's automaton alone. The evaluation's harness adds the strongest check. On every trial, an answer whose
supporting evidence lies wholly in the preserved suffix is checked by an executable landing assertion against the repair
the pristine corpus supplies, and every such check in the campaign passed.

\section{Evaluation}
\label{sec:eval}

The campaign is \code{tools/probes/src/recovery\_quality.cpp} of munch release \code{v1.6.0}, invoked as
\code{munch\_recovery\_quality 512 500 out.csv twitter.json 3}, with 512\,KiB of corpus per generated row, 500 trials
per cell and three independent seeds. Every seed is fixed, so the experiment is deterministic. The full CSV, the probe's
own summary, the five generated corpora, two checked-in analysis programs that derive this section's CSV-borne figures
from the archive, and a pinned reproduction script with the artifact hashes are deposited as a standalone record at
\doi{10.5281/zenodo.22178507}. The CSV has twenty-eight columns: the trial's failure offset, corruption end, first
mapped boundary, repairability and minimal-repair length, the decider's direct answer at the blind anchor, then per arm
the first and terminal positions with landing flags, the evidence interval and kind, the terminal outcome, attempts, the
per-move covered and landed counts, and the convergence triple. The pristine-oracle pass and the within-cell repeat
count are the harness summary's own attestations, archived beside the CSV and its per-move file. Every campaign
statistic and mapped-oracle count this section reports recomputes from that deposit, the pinned munch source tree the
script checks out and verifies against the bundled harness snapshot, and the pinned simdjson corpus the script fetches
and holds to its recorded digest. The real-document row reads the simdjson benchmark corpus \code{twitter.json} verbatim
($631{,}515$ bytes, byte-identical to \code{jsonexamples/twitter.json} at the simdjson repository's release
v3.10.1~\citep{simdjsondata}, SHA-256 \code{30721e49}\ldots, the full digest in the data notes and verified by the
reproduction script), held to the same complete-tokenizability assertion, the same damage protocol, and the same oracle
as the generated rows.

Quality, not throughput. A pristine corpus $x$, completely tokenizable by its row's grammar and asserted so, is damaged
at a position by one of three operations at $k \in \{1,4,16\}$: substituting $k$ bytes with pseudo-random bytes,
deleting $k$ existing bytes, or inserting $k$ pseudo-random bytes. Every operation leaves a suffix intact, in that the
damaged input $y$ satisfies $y[e..] = x[c..]$ for a corruption end $e$ and pristine anchor $c$ the operation designates.
For the damage start $d$, substitution sets $(e, c) = (d+k,\, d+k)$, deletion $(e, c) = (d,\, d+k)$, and insertion
$(e, c) = (d+k,\, d)$, so images shift by zero, $-k$, and $+k$ respectively; $e$ is the designated start of the suffix
the operation guarantees untouched, and a replacement byte that happens to equal the original never moves it earlier.
The oracle is the pristine boundary set mapped through that shift, exact for the theorem-predicted consequence on the
exhibited pristine repair, boundaries inside the damaged window having no image.

Six rows are measured: the conventional C-like grammar with strings and line comments, the same with block comments
alone, byte-level RFC 8259 JSON~\citep{rfc8259}, the bare C-like row, whose operator and punctuation bytes all certify
exactly, the split-friendly variant, which adds a certified newline byte to live windows, and the JSON grammar over the
real-world document above. The JSON grammar assumes UTF-8 well-formedness without checking it, so a corruption-injected
raw byte is absorbed inside strings that valid interchange would exclude. Trials the grammar absorbs without a scan
failure are counted and set aside, $37{,}264$ of the campaign's $81{,}000$ damage draws, leaving $43{,}736$ broken
scans. Damage positions are drawn by unbiased rejection sampling from the full span, clear of both corpus edges by at
least sixty-four bytes, independently per seed and per row, no two rows sharing a schedule or a payload stream; $43$
draws repeated an earlier position within their own cell and were included in the counts.

Eleven arms run under one completed-incident driver. After every failure the arm proposes a resume, the scan continues
from it, and the incident ends at the end of input, at a refusal, or at a budget of one hundred attempts, the terminal
outcome recorded per trial. All blind arms search from one past the failure offset, the same progress contract the
driver keeps; the two oracle arms floor their search at the corruption end, modeling a caller told the
operation-designated start of the guaranteed-untouched suffix, and are taken up with the trust discussion. Returned
offsets are byte-exact. The arms are certified recovery, the walk's answer with its evidence interval archived per
trial; certified-clean, the same walk under the oracle floor; exact, the library's anchored decider
(\code{next\_anchored\_start}, exact for complete-repair invariance on non-nullable sets by Theorem~\ref{thm:labels})
run as a procedure, the anchor advancing past a tail the decider refuses until a certificate holds; exact-clean, the
decider under the oracle floor; and skip-one, the search start itself, the lexical panic baseline. The rest are
delimiter conventions. The raw ones run in both placements, newline-past and semicolon-past answering one past the
delimiter's offset (the archive records them under the names newline and semicolon) and newline-at and semicolon-at at
the delimiter's own offset, with an in-harness regression failing if the two are interchanged. One past a final
delimiter is the end-of-input offset, which counts as a completed resume, so the placements answer together everywhere
and the paired regression asserts they differ by exactly the delimiter on every trial. The fresh-restart token-filtered
pair, token-newline and token-semicolon, makes the classical two-phase reading concrete at the lexical layer, skipping
until a fresh scan makes progress, discarding its emitted tokens through the first designated synchronizer, the
delimiter's own punctuation token or an all-whitespace token carrying the newline, and resuming one past it. The filter
is relative to the restarted scan, not the pristine one, and a restart inside an original string or comment can
reclassify interior text, so context blindness is reduced, never abolished.

Positions throughout are byte offsets, named apart: $d$ the damage start (the archive's \code{p} column), $e$ the
corruption end (the operation-designated start of the guaranteed-untouched suffix, as defined above),
$f = \mathrm{con}(y)$ the failure offset, the answer the resumed slice's first byte. An answer lands when it equals the
shift-mapped image of a pristine boundary outside the damaged window. Detection lag is real and measured. The signed lag
$f - d$ has median $0$ and 99th percentile $17$ over the range $[-436, 78]$, negative when a token opened before the
damage dies at it, and $f$ lies before $e$ on $40{,}401$ of the $43{,}736$ broken-scan trials. The metrics per arm and
trial are the first answer's position and landing, the terminal position and its landing where it lies inside the input,
the terminal outcome (completed, refused, or capped), attempts per incident, and mean signed overshoot from the first
mapped boundary at or past the corruption end. The tables' overshoot column reports mean absolute overshoot. For
completed incidents, convergence is the signed distance from the corruption end to where the resumed token-boundary
stream and the mapped pristine boundary stream agree forever after, token starts compared and never token identities,
with the mapped boundaries from the corruption end to that point counted lost and the non-landing emitted starts in the
divergence region counted spurious. Landing and overshoot score each arm's first answer per incident, skip-one's first
one-byte step included; repeated skipping surfaces in the attempts column.

Every trial is additionally stratified by repairability at the blind anchor, the routine-reported verdict of the shipped
\code{minimal\_repair} with every returned repair witness-verified by scanning, so an answer on a tail no repair can
complete is labeled as such while remaining in the pooled averages. Positive labels carry their witnesses; negative
labels rest on Theorem~\ref{thm:labels}, by which, over a non-nullable token set, for every input $y$ and every anchor
$0 \le c \le |y|$, $\code{minimal\_repair}(\suffix{y}{c}) = \code{nullopt}$ implies $\Rcomp(y, c) = \varnothing$. The
campaign's release refused nullable sets, whether or not a repair exists, but every campaign automaton is non-nullable,
so the theorem covers every labeled trial. Pooled rates below carry Wilson 95 percent intervals as descriptive
conditional-on-draw summaries, never as inferential bands.

\begin{proposition}[The pristine corpus is a repair]
\label{prop:oracle}
Let $y[e..] = x[c..]$ with $x$ completely tokenizable, and let $p \ge e$ be a repair-invariant resynchronization point
for $y$ anchored at $e$, relative to evidence $\window{y}{q}{q+w}$ with $e \le q$. Then $c + (p - e)$ is a token
boundary of $x$.
\end{proposition}

\begin{proof}
$x = x[0..c) \cdot x[c..]$ exhibits $x$ itself as a tokenizable repair of $y$'s preserved suffix, with $r = x[0..c)$; a
tokenizable repair is evidence-reaching for the stated evidence, so the definition places a token boundary at
$|r| + (p - e) = c + (p - e)$.
\end{proof}

\begin{figure}[!htbp]
\centering
\begin{tikzpicture}[
    scale=0.70,
    x=1cm, y=1cm,
    axisline/.style={draw=black, line width=0.5pt, -latex},
    tick/.style={draw=black, line width=0.5pt},
    damage/.style={draw=black, line width=3pt},
    boxline/.style={draw=black, line width=0.5pt},
    lead/.style={draw=black, line width=0.4pt, densely dotted},
    brace/.style={draw=black, line width=0.5pt, decorate,
                  decoration={brace, amplitude=4pt, raise=1pt}},
    bracem/.style={draw=black, line width=0.5pt, decorate,
                   decoration={brace, mirror, amplitude=4pt, raise=1pt}},
    scan/.style={draw=black, line width=0.4pt, -latex},
  ]
  \begin{scope}[yshift=0cm]
    \fill[black!12] (7.0,-0.16) rectangle (9.0,0.62);
    \draw[axisline] (0,0) -- (13.2,0);
    \node[anchor=west, font=\footnotesize] at (13.3,0) {bytes};
    \draw[damage] (2.0,0) -- (5.0,0);
    \draw[boxline] (9.0,-0.16) -- (7.0,-0.16) -- (7.0,0.62) -- (9.0,0.62);
    \draw[boxline, densely dashed] (9.0,0.62) -- (9.0,-0.16);
    \foreach \x in {2.8, 3.3, 5.0, 7.0, 9.0} { \draw[tick] (\x,-0.18) -- (\x,0.18); }
    \fill[black] (8.0,0) circle (2.2pt);
    \node[anchor=south, font=\small, inner sep=2pt] at (8.0,0.08) {$p$};
    \draw[brace] (7.0,0.72) -- (9.0,0.72);
    \node[anchor=south, font=\small] at (8.0,0.89) {$w$};
    \draw[scan] (3.3,-0.26) -- (6.85,-0.26);
    \node[anchor=north, font=\small] at (2.8,-0.46) {$f$};
    \node[anchor=north, font=\small] at (5.0,-0.46) {$e$};
    \node[anchor=north, font=\small] at (7.0,-0.46) {$q$};
    \node[anchor=north, font=\small] at (9.0,-0.46) {$q+w$};
    \draw[lead] (3.3,-0.32) -- (3.3,-0.98);
    \node[anchor=north, font=\small] at (3.3,-0.98) {$f{+}1$ (anchor)};
    \draw[bracem] (2.0,-1.80) -- (5.0,-1.80);
    \node[anchor=north, font=\footnotesize] at (3.5,-2.00) {damage span};
    \node[anchor=west, font=\small] at (0,0.89) {covered: $q \geq e$};
  \end{scope}
  \begin{scope}[yshift=-4.1cm]
    \fill[black!12] (4.0,-0.16) rectangle (6.0,0.62);
    \draw[axisline] (0,0) -- (13.2,0);
    \node[anchor=west, font=\footnotesize] at (13.3,0) {bytes};
    \draw[damage] (2.0,0) -- (5.0,0);
    \draw[boxline] (6.0,-0.16) -- (4.0,-0.16) -- (4.0,0.62) -- (6.0,0.62);
    \draw[boxline, densely dashed] (6.0,0.62) -- (6.0,-0.16);
    \foreach \x in {2.8, 3.3, 4.0, 5.0, 6.0} { \draw[tick] (\x,-0.18) -- (\x,0.18); }
    \fill[black] (5.4,0) circle (2.2pt);
    \node[anchor=south, font=\small, inner sep=2pt] at (5.4,0.08) {$p$};
    \draw[brace] (4.0,0.72) -- (6.0,0.72);
    \node[anchor=south, font=\small] at (5.0,0.89) {$w$};
    \draw[scan] (3.3,-0.26) -- (3.85,-0.26);
    \node[anchor=north, font=\small] at (2.8,-0.46) {$f$};
    \node[anchor=north, font=\small] at (4.0,-0.46) {$q$};
    \node[anchor=north, font=\small] at (5.0,-0.46) {$e$};
    \node[anchor=north, font=\small] at (6.0,-0.46) {$q+w$};
    \draw[lead] (3.3,-0.32) -- (3.3,-0.98);
    \node[anchor=north, font=\small] at (3.3,-0.98) {$f{+}1$ (anchor)};
    \draw[bracem] (2.0,-1.80) -- (5.0,-1.80);
    \node[anchor=north, font=\footnotesize] at (3.5,-2.00) {damage span};
    \node[anchor=west, font=\small] at (0,0.89) {uncovered: $q < e$};
    \node[anchor=west, font=\footnotesize] at (6.5,0.40)
      {the evidence begins before the preserved suffix};
  \end{scope}
\end{tikzpicture}
\caption{The coverage split on one byte axis: failure offset $f$, blind anchor $f{+}1$, corruption end $e$, evidence
$[q, q{+}w)$ with the answer $p$ inside it. Above, covered evidence begins at or past $e$ and
Proposition~\ref{prop:oracle} applies at anchor $e$. Below, uncovered evidence begins before the corruption end, the
anchor's $e \le q$ premise is unmet, and only the walk's guarantee at the blind anchor remains. The drawing is schematic
and not to scale, and its overlap with the damage span is only one uncovered geometry, since uncovered means $q < e$
alone.}
\label{fig:interval}
\end{figure}
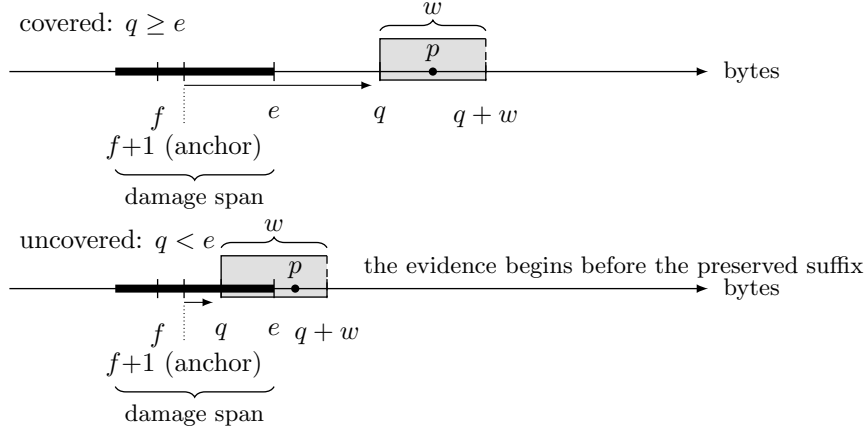

The harness's landing check rests on this proposition, and Figure~\ref{fig:interval} draws the two cases it separates on
one byte axis. A certified answer whose supporting evidence, the byte itself or the whole window occurrence, lies at or
after $e$ is repair-invariant for $y$ anchored at $e$, so its image must be a boundary of the pristine corpus, and the
harness asserts that on every such trial, failing the run on any violation. It asserted the landing on every one of the
$38{,}141$ evidence-covered first answers in this campaign, and on every covered answer of every later recovery move
besides, $40{,}885$ covered moves of the $48{,}217$ the certified incidents made in total, all landing. Answers whose
evidence begins before $e$ fall outside Proposition~\ref{prop:oracle}'s applicability at anchor $e$, its $e \le q$
premise unmet, so they are reported measurements, never counterexamples. The campaign measured $5{,}595$ such first
answers of which $1{,}589$ landed, reported and not asserted.

\begin{table}[htbp]
\centering
\small
\renewcommand{\arraystretch}{1.25}
\begin{tabular}{@{}>{\raggedright\arraybackslash}p{0.20\linewidth}p{0.27\linewidth}p{0.25\linewidth}p{0.20\linewidth}@{}}
\toprule
Result & Needs & Gives & Leaves open \\
\midrule
Resynchronization theorems & a certified byte or window occurrence at or after the anchor (for windows, flat
non-nullable sets), its evidence preserved, the repair's scan committing through it & a committed boundary at the image,
in every such repair & whether any such repair exists, which Theorem~\ref{thm:finite} decides from the tail \\
\addlinespace[2.5pt]
Corollary~\ref{cor:complete} & the same certificate, with the repair under consideration tokenizing whole &
the same boundary & existence again, decided the same way \\
\addlinespace[2.5pt]
Proposition~\ref{prop:collapse} & $p$ repair-invariant relative to its evidence, $c \le q$,
resumed suffix completely tokenizable & $\Rcomp = \Rreach$ at that anchor & whether either set is
inhabited, read off $E$ and $F$ in Theorem~\ref{thm:profile} \\
\addlinespace[2.5pt]
Proposition~\ref{prop:oracle} & $p$ repair-invariant at anchor $e$ with $y[e..] = x[c..]$, $x$ completely
tokenizable, evidence at or after $e$ & the image is a boundary of $x$ & pristine transfer when
$q < e$ \\
\addlinespace[2.5pt]
Corollary~\ref{cor:dichotomy} & $p$ repair-invariant relative to its evidence, $c \le q$ &
$\Rcomp \subsetneq \Rreach$ exactly when $\Rcomp = \varnothing$ and $\Rreach \ne \varnothing$ &
which of the two holds at a given anchor, which Corollary~\ref{cor:profile-complete} decides \\
\addlinespace[2.5pt]
Lemma~\ref{lem:overhang} & the walk's floor at the blind anchor & covered exactly when the evidence travel
reaches the overhang less one & how far the evidence lies from the anchor \\
\addlinespace[2.5pt]
Theorem~\ref{thm:labels} & the routines as shipped, on non-nullable sets in the campaign's release and on any token set
since v2.0.0 & $\Rcomp = \varnothing$ on negative-labeled tails, a shortest complete repair on positive ones & nothing,
the proof being Section~\ref{sec:exact}'s \\
\addlinespace[2.5pt]
Theorem~\ref{thm:labels}, the decider & the same & exact complete-repair decisions at the anchor, the
least invariant position & evidence-reaching exactness, which Theorem~\ref{thm:finite} supplies instead \\
\bottomrule
\end{tabular}
\renewcommand{\arraystretch}{1.0}
\caption{The guarantee boundary, read top to bottom as a ladder. Each row names one result, what it needs before it
applies, what it gives once it does, and what it leaves open, so a reader can locate any campaign figure on the exact
rung that backs it and see in the last column where that rung's guarantee stops. The campaign's figures cite their rows.
The $40{,}885$ covered move-level landings are Proposition~\ref{prop:oracle} instances, the $13{,}522$ collapses are
Proposition~\ref{prop:collapse} instances whose emptiness Theorem~\ref{thm:labels} supplies, and the $1{,}589$ uncovered
first-answer landings remain resynchronization-theorem instances at their blind anchors whose pristine transfer no row
supplies, so there they are measurements.}
\label{tab:ladder}
\end{table}

Near the seam the damaged input's true segmentation can genuinely diverge from the mapped pristine one. The same
conservative oracle is applied uniformly to every arm, though their near-seam exposure differs. Before any corruption, a
pristine-corpus oracle pass asserted that every certified answer from $512$ rejection-sampled offsets per row is a
boundary, with zero violations over the six rows, the same discipline the split-points paper uses. The sampler draws
without bias from every offset but the corpus's final byte.

Table~\ref{tab:ladder} draws the guarantee boundary once, so no later figure has to carry its own caveat. The table
lists the results the campaign consumes, with what each needs, what it gives, and what it leaves open.

\begin{table}[tbp]
\centering
\small
\setlength{\tabcolsep}{5pt}
\renewcommand{\arraystretch}{1.15}
\begin{tabular}{lrrrr}
\toprule
Row & Answers & Covered & Uncovered & Uncovered landed \\
\midrule
C-like, conventional & $7{,}356$ & $6{,}995$ & $361$ & $119$ \\
C-like, block comments & $4{,}754$ & $4{,}753$ & $1$ & $1$ \\
JSON, generated & $8{,}888$ & $8{,}051$ & $837$ & $554$ \\
C-like, split-friendly & $7{,}306$ & $6{,}817$ & $489$ & $291$ \\
C-like, bare & $7{,}877$ & $4{,}387$ & $3{,}490$ & $376$ \\
JSON, real document & $7{,}555$ & $7{,}138$ & $417$ & $248$ \\
\midrule
Total & $43{,}736$ & $38{,}141$ & $5{,}595$ & $1{,}589$ \\
\bottomrule
\end{tabular}
\renewcommand{\arraystretch}{1.0}
\caption{Evidence coverage per row, certified first answers only. Covered answers, evidence at or past the corruption
end, pass the asserted landing check, every one. Uncovered answers sit outside the transfer premise, and their landings
are reported, never asserted.}
\label{tab:coverage}
\end{table}

\begin{table}[tbp]
\centering
\small
\setlength{\tabcolsep}{5pt}
\renewcommand{\arraystretch}{1.15}
\begin{tabular}{lrrrrrr}
\toprule
Row & Answers & Refusals & Landing & Overshoot & Attempts & Conv. \\
\midrule
C-like, conventional & $7{,}356$ & $0$ & $96.7\%$ & $25.3$ & $1.02$ & $28$ \\
C-like, block comments & $4{,}754$ & $0$ & $100.0\%$ & $164.9$ & $1.00$ & $169$ \\
JSON, generated & $8{,}888$ & $0$ & $96.8\%$ & $4.2$ & $1.03$ & $7$ \\
C-like, split-friendly & $7{,}306$ & $0$ & $97.3\%$ & $24.3$ & $1.02$ & $27$ \\
C-like, bare & $7{,}877$ & $0$ & $60.5\%$ & $4.3$ & $1.47$ & $3$ \\
JSON, real document & $7{,}555$ & $0$ & $97.8\%$ & $3.3$ & $1.02$ & $26$ \\
\bottomrule
\end{tabular}
\renewcommand{\arraystretch}{1.0}
\caption{Certified recovery pooled per row over the campaign: answers and refusals,
first-answer landing rate, mean absolute overshoot in bytes, mean recovery attempts per incident, and mean signed
convergence distance in bytes from the corruption end to where the resumed boundary stream agrees with the mapped
pristine one forever after. The bare row's $60.5\%$ is its pooled landing rate. The block-comment row's $164.9$ bytes
accompanies the certificate sparsity the text discusses, and its convergence decomposes into the same distance once,
with no accumulation over attempts. Every figure recomputes from the archived CSV through the checked-in analysis
programs.}
\label{tab:pooled}
\end{table}

\paragraph{Three patterns in the recovery results.} First, the pristine-repair transfer is assertable exactly where its
evidence-preservation condition holds. The harness checked every evidence-covered answer against its mapped pristine
boundary, first answers and every later move alike, and all $40{,}885$ move-level checks passed, the $38{,}141$
first-answer checks among them (Table~\ref{tab:coverage} splits the coverage per row). Overall, certified recovery
landed $90.8$\% of its $43{,}736$ first answers, Wilson interval $[90.6, 91.1]$, refusing nothing in this campaign, and
the pooled figure's entire distance from the window-dense rows' $96.7$--$100$\% (Table~\ref{tab:pooled}) is one row's
trusted-anchor price, taken up below. The classical conventions carry no assertion clause anywhere. Skip-one lands
$10.0$\% pooled at $10.9$ recovery attempts per incident against the other answering arms' $1.0$--$1.1$ (the semicolon
variants average below one, their frequent first-move refusals contributing zero attempts), and it alone exhausts the
attempt budget, $776$ incidents.

Second, where certificates are dense the guarantee is also near. On the generated JSON row certified recovery's mean
absolute overshoot is $4.2$ bytes against newline-past's $29.5$ (Table~\ref{tab:quality} pools every non-oracle arm per
row), and its resumed stream agrees with the mapped pristine one within a median of $6$ bytes past the corruption.

Third, on the rows where certificates are scarce the guarantee is paid for in distance. On the block-comment row the
certified answers land $100$\% but $165$ bytes out on mean absolute overshoot, while newline-past lands $92.5$\% at
$22.5$ bytes, under a seventh of the distance (Figure~\ref{fig:landing}). The certificate guarantees soundness at its
own position, not proximity, and on these rows the price tracks certificate sparsity.

One pooled number, the first-answer landing rate, favors a classical convention over certified recovery. The
fresh-restart token-filtered newline convention lands $97.1$\% pooled, above certified recovery's $90.8$\%. Most of the
pooled gap is in the bare row. The coverage split above stratifies certified recovery's landing rate there. The
convention supplies no repair-universal boundary guarantee behind any of those landings and no answer at all where its
delimiter never occurs as a designated token, as the semicolon variants' refusal counts show at scale. Its median
convergence distance is $26$ bytes, more than twice certified recovery's $10$.

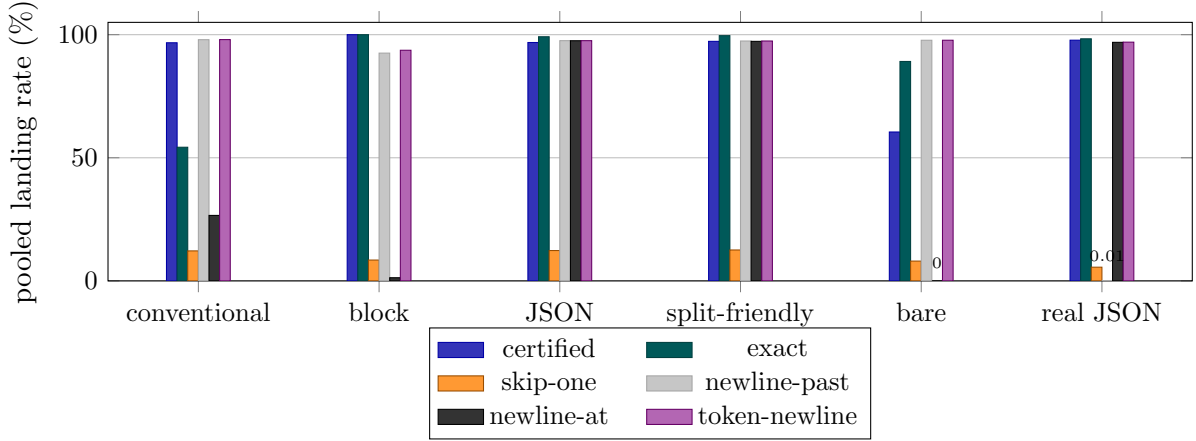
\begin{figure}[!htb]
\centering
\begin{tikzpicture}
\begin{axis}[
  ybar=0pt, bar width=4pt, width=\textwidth, height=5cm, area legend,
  ymin=0, ymax=105, ylabel={pooled landing rate (\%)},
  symbolic x coords={conventional, block, JSON, split-friendly, bare, real JSON},
  xtick=data, x tick label style={font=\small},
  legend style={at={(0.5,-0.18)}, anchor=north, legend columns=2, font=\small,
    /tikz/every even column/.append style={column sep=12pt}},
  ymajorgrids, tick label style={font=\small},
]
\addplot[fill=blue!65!black!80, draw=blue!65!black] table[col sep=comma, x=grammar, y=certified] {data/r6/r6-landing-figure.dat};
\addplot[fill=teal!70!black, draw=teal!50!black] table[col sep=comma, x=grammar, y=exact] {data/r6/r6-landing-figure.dat};
\addplot[fill=orange!80, draw=orange!60!black] table[col sep=comma, x=grammar, y=skip-one] {data/r6/r6-landing-figure.dat};
\addplot[fill=gray!45, draw=gray!70] table[col sep=comma, x=grammar, y=newline] {data/r6/r6-landing-figure.dat};
\addplot[fill=black!80, draw=black] table[col sep=comma, x=grammar, y=newline-at] {data/r6/r6-landing-figure.dat};
\addplot[fill=violet!60, draw=violet!80!black] table[col sep=comma, x=grammar, y=token-newline] {data/r6/r6-landing-figure.dat};
\legend{certified, exact, skip-one, newline-past, newline-at, token-newline}
\node[font=\tiny, rotate=90, anchor=west] at ([xshift=6pt]axis cs:block,2.3) {$1.3$};
\node[font=\tiny, rotate=90, anchor=west] at ([xshift=6pt]axis cs:bare,1) {$0$};
\node[font=\tiny, rotate=90, anchor=west] at ([xshift=2pt]axis cs:real JSON,1) {$0.01$};
\end{axis}
\end{tikzpicture}
\caption{Pooled first-answer landing rate per row and arm, recomputed from the archived campaign. The raw placements
still flip between near-perfect and near-zero across rows. Token filtering removed that failure mode in this campaign,
landing $97.0\%$ on the real document where its raw counterpart fails. Certified recovery needs no placement convention,
its one low row the trusted-anchor precondition. Blind exact wins that row and loses the conventional one, the
complementarity the text takes up. The semicolon arms are in Table~\ref{tab:quality}.}
\label{fig:landing}
\end{figure}

\paragraph{Short evidence is exposed: the byte path under damage.} The bare C-like row, where every operator and
punctuation byte certifies exactly, exercises the byte path at scale. It also lands worst, $60.5$\% overall, Wilson
interval $[59.4, 61.5]$. The covered split decomposes that rate within the one row and campaign schedule, grouped by the
returned certificate's kind, an observed, unrandomized split. Of the row's $2{,}408$ byte-evidence answers, $2{,}213$
rest on evidence not wholly in the preserved suffix, $91.9$ percent, and none of them land, while of its $5{,}469$
window answers $1{,}277$ are uncovered, $23.3$ percent, and $376$ of those uncovered answers land. All $4{,}387$
evidence-covered answers of the row pass the asserted landing check every time. The contrast does not lean on the
cross-row comparison, where evidence length is confounded with grammar and corpus. One explanation consistent with this
split is that a single byte is far easier for damage to produce than a two-to-four-byte occurrence with the right
context.

\begin{table}[bp]
\centering
\footnotesize \setlength{\tabcolsep}{2.8pt}
\renewcommand{\arraystretch}{0.87}
\makeatletter
\begin{tabular}{llrrrrrr}
\toprule
Row & Arm & Answers & Landing & Overshoot & Conv. & Completed & Refused \\
\midrule
C-like, strings and line comments & certified & $7{,}356$ & $96.7\%$ & $25.3$ & $28$ & $100.0\%$ & $0$ \\
 & exact & $7{,}356$ & $54.3\%$ & $25.2$ & $28$ & $100.0\%$ & $0$ \\
 & skip-one & $7{,}356$ & $12.2\%$ & $9.2$ & $3$ & $100.0\%$ & $0$ \\
 & newline-past & $7{,}356$ & $98.0\%$ & $25.6$ & $28$ & $100.0\%$ & $0$ \\
 & newline-at & $7{,}356$ & $26.6\%$ & $24.7$ & $28$ & $100.0\%$ & $0$ \\
 & semicolon-past & $7{,}356$ & $98.1\%$ & $82.7$ & $86$ & $100.0\%$ & $0$ \\
 & semicolon-at & $7{,}356$ & $98.0\%$ & $81.7$ & $85$ & $100.0\%$ & $0$ \\
 & token-newline & $7{,}356$ & $98.0\%$ & $25.6$ & $28$ & $100.0\%$ & $0$ \\
 & token-semicolon & $7{,}356$ & $98.1\%$ & $82.8$ & $86$ & $100.0\%$ & $0$ \\
\midrule
C-like, block comments & certified & $4{,}754$ & $100.0\%$ & $164.9$ & $169$ & $100.0\%$ & $0$ \\
 & exact & $4{,}754$ & $100.0\%$ & $164.9$ & $169$ & $100.0\%$ & $0$ \\
 & skip-one & $4{,}754$ & $8.5\%$ & $10.8$ & $4$ & $100.0\%$ & $0$ \\
 & newline-past & $4{,}754$ & $92.5\%$ & $22.5$ & $26$ & $100.0\%$ & $0$ \\
 & newline-at & $4{,}754$ & $1.3\%$ & $21.7$ & $26$ & $100.0\%$ & $0$ \\
 & semicolon-past & $4{,}752$ & $97.7\%$ & $125.7$ & $133$ & $100.0\%$ & $2$ \\
 & semicolon-at & $4{,}752$ & $97.6\%$ & $124.8$ & $132$ & $100.0\%$ & $2$ \\
 & token-newline & $4{,}754$ & $93.7\%$ & $23.1$ & $26$ & $100.0\%$ & $0$ \\
 & token-semicolon & $4{,}752$ & $97.7\%$ & $125.7$ & $133$ & $100.0\%$ & $2$ \\
\midrule
JSON, RFC 8259 lexical forms & certified & $8{,}888$ & $96.8\%$ & $4.2$ & $7$ & $100.0\%$ & $0$ \\
 & exact & $8{,}888$ & $99.2\%$ & $1.9$ & $4$ & $100.0\%$ & $0$ \\
 & skip-one & $8{,}888$ & $12.3\%$ & $9.9$ & $12$ & $100.0\%$ & $0$ \\
 & newline-past & $8{,}888$ & $97.6\%$ & $29.5$ & $32$ & $100.0\%$ & $0$ \\
 & newline-at & $8{,}888$ & $97.5\%$ & $28.5$ & $31$ & $100.0\%$ & $0$ \\
 & semicolon-past & $208$ & $6.2\%$ & $8.9$ & $0$ & $0.2\%$ & $8{,}874$ \\
 & semicolon-at & $208$ & $0.0\%$ & $9.9$ & $--$ & $0.0\%$ & $8{,}888$ \\
 & token-newline & $8{,}888$ & $97.6\%$ & $29.5$ & $32$ & $100.0\%$ & $0$ \\
 & token-semicolon & $0$ & $--$ & $--$ & $--$ & $0.0\%$ & $8{,}888$ \\
\midrule
C-like, split-friendly & certified & $7{,}306$ & $97.3\%$ & $24.3$ & $27$ & $100.0\%$ & $0$ \\
 & exact & $7{,}306$ & $99.6\%$ & $24.7$ & $27$ & $100.0\%$ & $0$ \\
 & skip-one & $7{,}306$ & $12.6\%$ & $9.2$ & $3$ & $100.0\%$ & $0$ \\
 & newline-past & $7{,}306$ & $97.4\%$ & $25.3$ & $28$ & $100.0\%$ & $0$ \\
 & newline-at & $7{,}306$ & $97.3\%$ & $24.3$ & $27$ & $100.0\%$ & $0$ \\
 & semicolon-past & $7{,}304$ & $97.7\%$ & $81.8$ & $86$ & $100.0\%$ & $2$ \\
 & semicolon-at & $7{,}304$ & $97.6\%$ & $80.9$ & $85$ & $100.0\%$ & $2$ \\
 & token-newline & $7{,}306$ & $97.4\%$ & $25.3$ & $28$ & $100.0\%$ & $0$ \\
 & token-semicolon & $7{,}304$ & $97.8\%$ & $82.1$ & $86$ & $100.0\%$ & $2$ \\
\midrule
C-like, bare & certified & $7{,}877$ & $60.5\%$ & $4.3$ & $3$ & $100.0\%$ & $0$ \\
 & exact & $7{,}877$ & $89.1\%$ & $1.2$ & $3$ & $100.0\%$ & $0$ \\
 & skip-one & $7{,}877$ & $8.0\%$ & $8.4$ & $2$ & $100.0\%$ & $0$ \\
 & newline-past & $7{,}877$ & $97.7\%$ & $22.1$ & $24$ & $100.0\%$ & $0$ \\
 & newline-at & $7{,}877$ & $0.0\%$ & $21.2$ & $24$ & $100.0\%$ & $0$ \\
 & semicolon-past & $7{,}877$ & $97.7\%$ & $74.9$ & $78$ & $100.0\%$ & $0$ \\
 & semicolon-at & $7{,}877$ & $97.6\%$ & $73.9$ & $77$ & $100.0\%$ & $0$ \\
 & token-newline & $7{,}877$ & $97.7\%$ & $22.1$ & $24$ & $100.0\%$ & $0$ \\
 & token-semicolon & $7{,}877$ & $97.7\%$ & $74.9$ & $78$ & $100.0\%$ & $0$ \\
\midrule
JSON, real-world document & certified & $7{,}555$ & $97.8\%$ & $3.3$ & $26$ & $100.0\%$ & $0$ \\
 & exact & $7{,}555$ & $98.3\%$ & $2.2$ & $25$ & $100.0\%$ & $0$ \\
 & skip-one & $7{,}555$ & $5.6\%$ & $48.3$ & $9$ & $89.7\%$ & $0$ \\
 & newline-past & $7{,}555$ & $0.0\%$ & $18.4$ & $49$ & $100.0\%$ & $0$ \\
 & newline-at & $7{,}555$ & $96.9\%$ & $17.5$ & $40$ & $100.0\%$ & $0$ \\
 & semicolon-past & $708$ & $0.7\%$ & $16830.4$ & $3$ & $0.2\%$ & $7{,}543$ \\
 & semicolon-at & $708$ & $0.0\%$ & $16829.9$ & $--$ & $0.0\%$ & $7{,}555$ \\
 & token-newline & $7{,}555$ & $97.0\%$ & $26.8$ & $49$ & $100.0\%$ & $0$ \\
 & token-semicolon & $0$ & $--$ & $--$ & $--$ & $0.0\%$ & $7{,}555$ \\

\bottomrule
\end{tabular}
\makeatother
\caption{\small Every non-oracle arm pooled per row, three seeds of five hundred trials per cell, the body generated by
the archived analysis program and taken up verbatim. The columns are first-answer count and landing rate, mean absolute
overshoot, mean signed convergence distance, completed share, and terminal refusals, incidents that never answered or
answered and then starved, defined with the metrics above. The two oracle arms are omitted as constant, $100$ percent
landing on every row, asserted per trial. The per-cell grid over three operations, three damage sizes, and three seeds
is in the archived CSV, where the past placements appear under the names newline and semicolon. Elsewhere the rows
shorten, in this order, to conventional (strings and line comments), block, generated JSON, split-friendly, bare, and
real JSON. One rounding is load-bearing. The real document's newline-past cell prints $0.0$ from exactly one landing in
$7{,}555$, while the bare newline-at cell is a true zero.}
\label{tab:quality}
\end{table}

\paragraph{The mechanism behind coverage, from the archive.} Coverage has a proved part and an observed part. The
identity is proved algebraically, and the archive recomputes the observed bins (Figure~\ref{fig:overhang}) while
checking every archived row against it.

\begin{figure}[!htb]
\centering
\begin{tikzpicture}
\begin{axis}[
  width=0.85\textwidth, height=4.0cm,
  ymin=0, ymax=105, ylabel={rate (\%)},
  xlabel={overhang $e-f$, binned},
  xtick={1,2,3,4,5,6,7,8},
  xticklabels={$\le 0$, $1$, $2$, $3$, $4$--$7$, $8$--$15$, $16$--$31$, $32+$},
  tick label style={font=\small}, label style={font=\small},
  legend style={font=\small, at={(0.98,0.05)}, anchor=south east},
  ymajorgrids,
]
\addplot[mark=*, mark size=1.6pt, draw=black] table[x=bin, y=covered] {data/r6/r6-overhang.dat};
\addplot[mark=square*, mark size=1.5pt, draw=black!50, mark options={fill=black!12}]
  table[x=bin, y=landed] {data/r6/r6-overhang.dat};
\legend{covered, landed}
\end{axis}
\end{tikzpicture}
\caption{Coverage and first-answer landing for certified recovery against the damage's overhang past the failure offset,
recomputed from the archived CSV. Bin populations, left to right: $3{,}335$, $8{,}980$, $1{,}133$, $2{,}733$, $9{,}708$,
$5{,}103$, $11{,}558$, and $1{,}186$ answers. At overhang at most one the anchor shelters the evidence and coverage is
forced, and every uncovered answer lives to the right.}
\label{fig:overhang}
\end{figure}
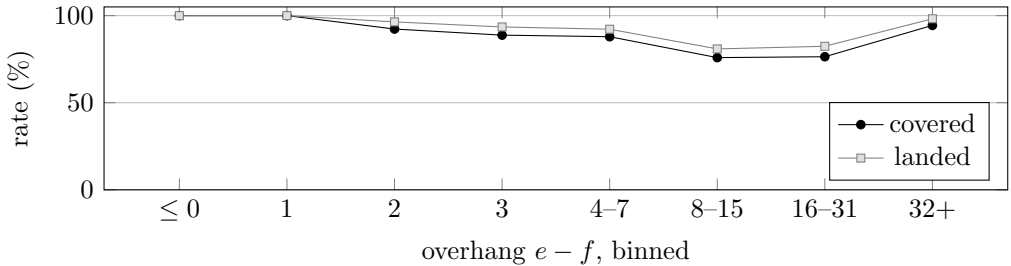

\begin{lemma}[The overhang law]
\label{lem:overhang}
Call $\eta = e - f$ the damage's \emph{overhang} and $\ell = q - (f{+}1)$ the walk's \emph{evidence travel} for that
answer, the distance from the blind anchor to where the answer's evidence begins. The answer is covered, $q \ge e$,
exactly when $\ell \ge \eta - 1$; the search starts at the blind anchor, so $\ell \ge 0$, and an overhang of at most one
forces coverage.
\end{lemma}

\begin{proof}
$q \ge e$ rearranges to $q - (f{+}1) \ge (e - f) - 1$, and the walk's documented floor puts $q$ at or past $f{+}1$.
\end{proof}

The threshold is sharp. For every $\eta \ge 2$ choose $n$ with $2n \ge \eta + 2$, and in $x = (\code{ab})^n$ replace
$\window{x}{2}{2+\eta}$ by the length-$\eta$ prefix of $(\code{ba})^\omega$. The corruption end is the one the operation
designates, $e = 2 + \eta$, as everywhere else in this paper. This block is chosen because it also differs from $x$ at
every replaced position, so the damage is actual across the whole designated span. The instance stands under either
reading of where the damage ends, the designated end or the byte-difference end one past the last differing byte. Every
\code{a} of a completely tokenizable string over $\{\code{ab}\}$ begins a token, so \code{a} is a certified byte. The
scan commits the leading \code{ab} and fails at $f = 2$; the blind search from three returns the certified \code{a} at
$q = p = 3$, and $q < e$ for every $\eta \ge 2$. At every overhang at least two there is thus an uncovered certified
instance, and the lemma's threshold cannot be weakened.

\paragraph{The known-clean arms.} The effect of supplying a known-clean search floor is measured. An oracle-aided
campaign arm runs the same walk given the operation-designated start of the guaranteed-untouched suffix, starting at the
corruption end or one past the failure, whichever is later, and a second oracle arm runs the exact decider from the same
floor. Their soundness arguments differ, and each arm's landing is asserted per trial. The clean walk's evidence is
covered by construction and every answer lands under the pristine transfer, while the clean decider returns no evidence
and its landings rest on complete-repair invariance with the pristine corpus the complete repair. Each arm gives
$43{,}736$ first answers, the same count as the blind walk's, every one landing, with no refusals. The clean and blind
walks answer the same trials but not at the same places. Their first positions differ on $4{,}592$ of the $43{,}736$
paired answers, $10.5$ percent, so what the truth about the damage changes is both which theorems cover the answers and
where a tenth of them land, which makes the deployment precondition's content a measured column.

\paragraph{The real document.} The real-document row runs the same JSON grammar over twitter.json, real Twitter API
output. Certified recovery lands as well here as on generated JSON, $97.8$\% of its $7{,}555$ first answers against
$96.8$\%, at a mean absolute overshoot of $3.3$ bytes against $4.2$, though its mean signed convergence distance grows
from $7$ bytes to $26$. The newline baselines expose the placement sensitivity the campaign measures as
outcome-critical, which is why both placements are named. Newline-past collapses to one landing in $7{,}555$ first
answers, $0.01$\%, because the document is pretty-printed, almost every newline is followed by indentation, and skipping
past the newline lands inside the run. Retained, newline-at lands $96.9$\% here, by hitting the whitespace run the
newline begins. Neither number transfers, since on the generated rows the same one-byte change runs the other way,
pooled per row, $98.0$\%
falling to $26.6$\% on the conventional row, $92.5$\% to $1.3$\% on block comments, $97.7$\% to $0.0$\% on bare, because
the generated layout makes the newline mid-run and the byte after it a token start. The token-filtered form settles the
placement question by resuming at the end of a synchronizer token emitted by its restarted scan. Token-newline lands
$97.0$\% on this document and $93.7$--$98.0$\% on every generated row, because it never resumes inside a token its own
scan just emitted, so the token-filtered baseline both repairs the classical convention and removes its one-byte trap.
The landing swing between the two placements was a byte-level artifact, and what no delimiter reading escapes is
availability and the missing repair-universal boundary guarantee. The certificate's soundness is placement-free and, for
a fixed automaton with preserved evidence, layout-independent, while availability, coverage, refusal, and proximity
remain corpus and layout properties, measured above.

\paragraph{Refusal is explicit and scoped.} Certified recovery refused nowhere in this campaign, $43{,}736$ first
answers on $43{,}736$ broken scans. Its refusal mode is real, pinned by fixture, and unexercised on these corpora under
the rejection sampler. On the window-dense rows every answer is a window answer, those grammars carrying no useful byte
certificates. The bare and split-friendly rows supply the byte path above, $2{,}408$ and $51$ byte-evidence answers. The
delimiter arms show what refusal looks like without a theorem. The semicolon variants each refuse over fifteen thousand
of the campaign's $43{,}736$ trials at the first move, over sixteen thousand terminally once incidents that answered and
then starved are counted, raw and token-filtered alike and nearly all on the two JSON rows. The token-filtered form
returns nothing at all on the real document, whose single pristine semicolon lies inside a string token that every
restarted scan in this campaign emitted whole, while the raw forms' few answers lie at or immediately after that
in-string byte, and semicolon-past's answers land $0.7$ percent. Skip-one never refuses and lands rarely, $10.0$\%
pooled, the opposite failure.

\paragraph{The vacuity fence, stratified and then decided.} Every trial carries the verdict of the shipped
\code{minimal\_repair}, the routine's answer to whether any repair completes the blind tail. It labels $26{,}928$ trials
repairable against $16{,}808$ beyond repair, and the walk answered every one of the latter, $38.4$ percent of its first
answers. By Theorem~\ref{thm:labels} a negative label is exact, so at those $16{,}808$ blind anchors $\Rcomp$ is empty
and the complete-repair reading of Corollary~\ref{cor:complete} claims nothing; for $13{,}522$ of them the strengthened
semantics is proved to fare no better. Every certified answer's evidence lies at or after its blind search anchor, so
Proposition~\ref{prop:collapse} needs only one further premise per trial, a resumed scan that completes in one attempt.
This premise holds for $13{,}522$ of the $16{,}808$ trials. These comprise the $12{,}023$ whose evidence the damage
spared and $1{,}499$ more on evidence not wholly in the preserved suffix. At all $13{,}522$ blind anchors the two
quantifiers therefore coincide, and both are empty. Coincidence holds more widely still, because on the $26{,}928$
trials the routine labels repairable, every returned repair was witness-verified by scanning, so $\Rcomp$ is inhabited
by exhibition and Corollary~\ref{cor:dichotomy} forces the two sets equal there. Equality therefore holds at $40{,}450$
of the campaign's $43{,}736$ answers, $92.5$ percent, and only the $3{,}286$ remain candidates for strict inclusion at
all.

The remaining $3{,}286$, whose first resumed scans failed downstream, are the only answers this campaign leaves
undetermined in content at the blind anchor. Section~\ref{sec:exact} can decide each of them from its preserved tail. A
trial's negative label gives $E < m$, so by Corollary~\ref{cor:profile-complete} the inclusion is strict exactly when
$h \le E$, where $h = q + w - c$ is the evidence endpoint. Applying this test requires rerunning the campaign, because
the archive keeps each trial's offsets and labels but not its damaged input, so the count stands at $3{,}286$. These
answers' soundness stands either way.

Coverage did not enter that argument; what coverage decides is the transfer. The $12{,}023$ covered answers all landed,
guaranteed at the corruption-end anchor where the pristine corpus is a complete repair, the same guarantee the
known-clean arms exercise by construction, while the $1{,}499$ uncovered collapses carry no such transfer and land
$1{,}066$ times, which no theorem guarantees. For uncovered evidence the corruption-end anchor is no anchor at all, the
evidence beginning before it, outside the definition's domain, though each such answer stays governed by the walk's
soundness theorem at its own blind anchor. For deployment, then, every returned answer is sound at its blind anchor, and
coverage decides which answers additionally carry Proposition~\ref{prop:oracle}'s pristine transfer, the $12{,}023$
against the $1{,}499$. The failure offset cannot draw that line; the returned evidence interval can, which is why the
evidence-returning form of the interface returns it.

\paragraph{The anchored decider, a different-property comparator.} The second machine in the campaign is the library's
anchored decider, exact for complete-repair invariance by Theorem~\ref{thm:labels}, the campaign's release holding it on
non-nullable sets. With the tail's end known to be the end of input, it answers the positions every completely
tokenizable repair agrees on and refuses tails no repair completes. That is a different predicate from the certificates'
evidence-reaching soundness, and the two separate on a two-token example. Over $\{\code{ab}, \code{ba}\}$ the decider
answers position zero of the tail \code{ab}, yet relative to the one-byte evidence $\window{y}{0}{1}$ the repair
\code{b} commits \code{ba} and dies having reached that evidence's image, with no committed boundary at the answer's
image, so the decider's answer is not evidence-reaching-invariant for that interval. With the whole tail as evidence,
reaching and completing already coincide, which is the decider's own reading. The comparison below is therefore between
two different guarantees. The campaign binds them where their properties overlap and tests both directions on the
decider's direct call at the blind anchor. The direct answer is archived as a column. The direct call answered every one
of the $26{,}928$ routine-labeled-repairable trials at or before the walk, never later, which is not a coincidence of
these rows:

\begin{proposition}[The decider never answers later than the walk]
\label{prop:order}
Fix a non-nullable token set and an anchor $c$. Suppose the walk, searching from $c$, answers $p$ on evidence
$\window{x}{q}{q+w}$ with $c \le q$, and the anchored decider, consulted at the same anchor, which by
Theorem~\ref{thm:labels} reports the earliest position at or after $c$ meeting its certificate predicate, answers $d$.
Then $d \le p$.
\end{proposition}

\begin{proof}
Every completely tokenizable repair is evidence-reaching, so by the resynchronization theorems a committed token begins
at $p$'s image in every complete repair anchored at $c$, and the position $p$ is therefore one the complete repairs
agree on. By Theorem~\ref{thm:labels} the decider reports the earliest position at or after $c$ meeting that predicate,
so where it answers at all its answer $d$ is at or before every such position; $p$ is such a position, hence $d \le p$.
Exactness alone would not suffice, because an invariant position that is not the earliest could exceed the walk, so it
is the minimality Theorem~\ref{thm:labels} states, matched by the shipped \code{next\_anchored\_start}, that the bound
rests on.
\end{proof}

The measured containment is that proposition instantiated, worth $10{,}738$ bytes across those trials. On every trial
the routine labeled unrepairable, the decider's direct call refused. That agreement is only a consistency regression
between the two routines, since both walk the same scenario machinery and are exact by the same theorem,
Theorem~\ref{thm:labels}. Every repair the routine did return was witness-verified by scanning it against its tail.

Run as a procedure, the decider answers everything and lands $89.8$ percent pooled against the walk's $90.8$, and the
per-row split is the finding. The decider wins the bare row, $89.1$ against $60.5$ percent, and loses the conventional
row, $54.3$ against $96.7$. On that row an advanced anchor still sits inside a broken string and anchored reasoning
binds repairs of a damaged prefix, while the walk's occurrence-universal certificates land wherever their evidence
survived. Under the oracle floor both land on every trial, $100$ percent asserted, and each rests on its own result. The
clean walk rests on the pristine transfer, its evidence covered by construction, and the clean decider on
complete-repair invariance, with the pristine corpus as the complete repair, the contract Theorem~\ref{thm:labels}
proves. The two machines are complementary because their properties are. Anchored completeness paid off at this
campaign's trustworthy anchors and evidence-order search at its poisoned ones, two row-level contrasts and no law about
classes of anchors. At $13{,}522$ of the poisoned anchors the surviving answers carry no evidence-reaching content
either, only the clean-anchor transfer where their evidence is covered.

\paragraph{Convergence: what each convention costs the stream.} First landings decide where an arm resumes, and
convergence measures what that choice costs downstream. It is the signed distance from the corruption end to the point
where the resumed boundary stream agrees with the mapped pristine one forever after, taken over completed incidents, and
the boundaries from the corruption end to that point, the initial jump included, are counted lost. Skip-one, the worst
arm by pooled first-answer landing rate, converges at median $2$ bytes and loses only $0.90$ boundaries per error
pooled, at the price of $10.9$ attempts per incident, $0.9$ spurious starts per error inside the divergence region, and
$776$ budget exhaustions.

Certified recovery converges at a pooled median of $10$ bytes, $2$ on the bare row and $6$ on generated JSON. These
convergence figures are medians, whereas the convergence column of Table~\ref{tab:quality} reports means. Certified
recovery loses $6.85$ boundaries per error pooled, from $0.22$ on the bare row to $34.6$ on block comments. On the
block-comment row, where every completed certified incident converges at its first answer, the lost boundaries are the
initial jump the overshoot already stated. They are the same quantity counted in tokens instead of bytes, and they
appear once without accumulating. On the other rows a later move can lose boundaries the first answer did not skip.

The newline conventions lose slightly more than certified recovery, $8.01$ past and $7.42$ retained, at a median of
$24$--$26$ bytes. The semicolon forms lose a statement's worth regardless of placement, $24.4$, $23.4$, and $24.5$ for
past, retained, and token-filtered, at a median of about $60$. The flip that swings first landings between $0.01$ and
$96.9$ percent on the real document moves convergence by roughly a token, so the classical convention's real cost is the
discarded remainder of the line. The token-filtered semicolon form shows the availability cliff in stream terms. It
refuses every incident on the real document, because it declines to treat the document's single candidate delimiter,
which lies inside a string token, as a synchronizer.

\paragraph{The price of evidence order, measured.} The walk answers with the first certificate in evidence order, and
the campaign priced that fence. Of the campaign's $43{,}736$ certified first answers, $90$ were nonminimal, by $90$
bytes in total, one byte each. The fence is real and its cost on these rows is negligible, which is worth knowing before
anyone complicates the walk to remove it.

\paragraph{Independent bounded verification.} A from-scratch maximal-munch reference model, written in Python against
only the public API and kept in the munch test suite as the recovery cross-check, is checked against the shipped
library's own scan on $7.6$ million strings with zero disagreements. It also drives bounded exhaustive sweeps over $89$
small non-nullable literal-token sets, comparing boundaries and committed lengths, and each obligation carries its own
denominator. For walk soundness the sweeps cover $46{,}035{,}590$ (answer, repair) pairs, the evidence-reaching premise
holding on $1{,}308{,}265$ of them, and find no violations. The collapse proposition's sweeps cover $836{,}387$ walk
answers, its premise holding on $148{,}920$, again with no violations. For the anchored decider they cover
$4{,}712{,}940$ (answer, completing repair) pairs with no violations, and the decider is never later than the walk on
the $174{,}368$ tails where the bounded search exhibits a completing repair and both answered. None of the routine's
$95{,}836$ negative labels is refuted by any repair up to length four, a consequence of Theorem~\ref{thm:labels} kept as
a check on the implementation. For pristine transfer, an independently mapped damage sweep of $2.36$ million trials has
all $1{,}068{,}953$ covered answers landing, which replicates Proposition~\ref{prop:oracle} without the harness. Every
quantifier is bounded where the contracts are not, and literal tokens leave regex structure and priority ties outside
the tested universe, so these runs are evidence, never proof. The drivers' executed negative controls make a silent
check an exit-failing defect of the cross-check itself, which runs in the munch test suite, where the library it checks
is built.

\paragraph{Limitations and external validity.} The corpora are generated except for one real-world document, the damage
model is synthetic throughout, and real edit traces were not studied. Cell sizes are broken-scan counts and vary with
absorption, and while pooled rates carry Wilson intervals, individual cells remain small in the absorbed strata. Damage
positions come from unbiased rejection sampling over the full span under three independent seeds, with every row's
schedule and payload streams seeded apart. Per-seed figures are printed beside the pooled ones, and the $43$ within-cell
repeated draws are counted. Every arm runs to a terminal outcome under one driver, so the completed-incident view is a
reported column, and only skip-one ever exhausts the budget of one hundred attempts. The anchored procedure's per-move
cost was not priced against the walk's single forward scan, since the comparison here is of quality alone. That cost is
one scenario play per automaton state over the remaining tail, with a jump-table construction that can reach quadratic
in the tail and an advancing procedure that is cubic in the worst case. Every campaign automaton is non-nullable, so the
campaign release's nullable refusal is never exercised and the comparison never leaves its domain.

Six rows over five grammars were measured, and none of the quality figures generalize beyond them, since proximity,
refusal, and exposure rates varied with certificate density and evidence length, an association this campaign observed
but did not isolate. The repairability labels and the decider's refusals rest on the anchored routines' contracts,
proved in Section~\ref{sec:exact}, and the two routines share their scenario machinery, so their agreement in the
campaign shows consistency only. The $3{,}286$ undetermined answers are decidable per trial by the same section and were
not decided, because the archive carries no damaged inputs. The byte path is exercised at scale by one grammar row, with
$51$ further byte answers on the split-friendly row. The search is the shipped form, whose window component searches
lengths two through four (the byte search is not length-capped), and the cap's sensitivity was not varied. The seam
oracle is conservative and applied uniformly. No claim is made that every sound position is discoverable by the searched
certificates, and throughput was not measured. The convergence metric compares whole boundary streams, dropped tokens
included, so it credits agreement whether or not it is trivial, and the lost and spurious columns carry the divergence
region's content.

One grammar tried for the byte path produced no broken-scan trials. The grammar, a log-line lexer whose tokens are a run
of non-newline bytes and the newline, tokenizes every byte string, so no corruption can break its scan and recovery
never fires.

\section{Related work}

Panic mode is the textbook's simplest strategy at both levels. At the lexical level it deletes successive characters
until a well-formed token appears, and at the parsing level it discards symbols until one of a designated set of
synchronizing tokens is found, usually delimiters such as semicolon, selected by the compiler designer~\citep[pp.~88,
164--165]{dragon1986}. The guarantee attached is termination and, where a synchronizer is found, membership in the
designated set. No repair-universal boundary guarantee is attached for the position landed on, and the text notes that
lexical panic-mode recovery may confuse the parser. Our progress lemma reproduces the termination half of that classical
guarantee, membership following from the search rule where a delimiter convention is in force, and the certificates add
what that guarantee lacked, a theorem about the position landed on that ranges over repairs.

A close ancestor at the lexical level is Boullier and Jourdan's two-level scheme, which repairs table-driven scanners
and LR parsers alike. Correction models specify insertions, deletions, and replacements around the detection point,
validated by the analyzer accepting the model's trailing symbols, and when local repair fails the parser skips to a
grammar-writer key terminal while the scanner deletes the offending character~\citep{boullier1987repair}. The classical
comparator ideas have close antecedents here: the scanner's deletion at the offending character is the nearest
one-character antecedent of the skip-one convention, which advances from the committed failure offset instead; the
key-terminal skip anteceded the token-aware designated-token discipline without supplying a repair-universal boundary
guarantee; and the bounded correction models carry a per-correction success test, continued acceptance. The scheme is
the classical position at its most parameterized. It chooses repairs and fallback recoveries judged excellent, mean, or
poor against what a human reader would have corrected, and it states operational acceptance conditions but no
repair-universal boundary property of the resume position. Repair invariance adds the missing quantifier for a scoped
form of this design. Once an anchor and its evidence are fixed, every prefix-local correction that preserves the
evidence and whose scan commits through it agrees at the certified position.

The repair school selects one concrete fix and validates it forward. Burke and Fisher~\citep{burke1987repair} defer
parse actions so repairs can be tried left of the error token, generate single-token candidates at each trial point, and
judge every candidate by the distance it lets the parse advance before blocking. A candidate must advance at least one
token in their implementation to remain in consideration, the farthest-checking candidates are kept, and the pruning
criteria were arrived at largely through experimentation with erroneous Pascal and Ada programs. When simple recovery
fails, scope closers are tried at the same trial points under their own acceptance rules. Success means the correction a
knowledgeable human reader would choose, a notion the paper calls impossible to define precisely. The validation
criterion concerns the proposed repair's parse-ahead, never the position scanning resumes at, and nothing quantifies
over the repairs not chosen, which is the quantification repair invariance makes.

Nearest in spirit is Richter's noncorrecting recovery, which attempts no correction and analyzes the remainder of the
text against the subword language, buying provable properties about the error messages, no spurious reports and no
skipped text, with no assumptions about the nature of the errors~\citep{richter1985noncorrecting}. Two of its 1985
observations anticipate this paper's measurements: that panic mode trusts a single fiducial symbol ``to tell the parser
where it is'' even when ``a string of more than one symbol were less ambiguous'', which is the byte-versus-window
distinction the evaluation prices, and that the properties worth having are the provable ones. The deliverable still
differs the same way, because subword membership says the damaged remainder can extend to a sentence, which is
repairability at the language level, while repair invariance says where every such extension places a boundary.

The minimum-distance school fixes a metric and searches for minimum-cost repairs. Aho and
Peterson~\citep{aho1972minimum} parse any input to completion with the fewest possible errors, through a covering
grammar whose error productions mark the positions and kinds of the corrections, in time cubic in the input; they offer
the algorithm as the yardstick for evaluating recovery methods, and the textbook echoes that framing while pricing it
out of practice. The quantification is the dual of ours, in that minimum distance searches over all repairs to return a
minimum-error repair and its parse, while repair invariance asks what every repair reaching the certificate's evidence
agrees on, trading the richer deliverable for the stronger quantifier. The shortest reaching repair and shortest
counterrepair of Section~\ref{sec:exact} are shortest replacement prefixes at the lexical layer, found among $n + 1$
candidates, but they minimize the prefix's length, not an edit distance, so the parallel ends at the quantifier. The
school's modern member, Diekmann and Tratt's CPCT+, returns the whole optimum. It reports the complete set of
minimum-cost repair sequences at an error location, repairs $98.4$\% of $200{,}000$ real syntactically invalid Java
programs within a half-second timeout each, and uses that completeness to report fewer than half the error locations
panic mode does on the same corpus~\citep{diekmann2020dontpanic}. The deliverable is richer still, every best repair at
once; the quantification is unchanged, over repairs to select among, never over repairs to agree.

The incremental lexers cited here maintain consistency across edits without certifying recovery points across
hypothetical repairs. Wagner and Graham's general incremental lexer restarts the batch machine from saved states at
startable token boundaries and reproduces a batch scan of the incorporated text; errors are handled representationally,
by unmatched-text tokens and by error patterns recognizing near misses, and their history-based alternative leaves
invalid modifications unincorporated and retries them later~\citep{wagner-lexing}. Hugo and Hansson's divide-and-conquer
generator stores lexical errors in its intermediate results, and the thesis also describes a local-error variant that
lexes the following text from the starting state~\citep{hugo2015incremental}. These works target batch or sequential
agreement for the text they incorporate; none states that a chosen position is a committed boundary in every repair of
an unrepaired prefix whose scan commits through the stated evidence.

Coding theory offers a longstanding form of the resynchronization question. A synchronizing word for a prefix code
drives the decoder into a known state from every state on which it is defined, so a complete decoder realigns at each
intact occurrence in the stream an error leaves behind, and a partial one, under the cited definition, exactly where the
word is defined from the state the error left. Almost all complete, equivalently maximal, finite binary prefix codes
admit such words as the number of codewords tends to infinity~\citep{freiling2003sync}, and Ryzhikov and Szyku{\l}a
study the complexity of finding shortest ones, proving strong inapproximability for decoder-defined codes beside
algorithms for literal ones~\citep{ryzhikov2018sync}. Synchronizing words are the closest prefix-code analogue of
repair-independent reanchoring, and the certified byte is their lexical cousin. The general setting does not transfer
directly, since maximal-munch token sets need not be prefix codes, one token may prefix another, and commitment waits on
lookahead and rollback, a decoder model the prefix-code results do not address. The certificates quantify instead over
repairs of the broken prefix, evidence-reaching failed scans included, are decided statically per automaton, and refuse
where their bounded search certifies nothing. The coding-theory question is the existence and length of a universal
resynchronizer for a code; ours is the soundness of one position in one damaged input, under a scanning discipline whose
rollback desynchronizes the decoder picture.

Three axes separate the lines above: what a position is keyed to, what its guarantee ranges over, and what happens where
nothing is found. Table~\ref{tab:axes} places the classical convention, the lexical-level ancestor, the prefix-code
analogue, and the certificates on those axes, each row argued against its sources in the paragraphs above.

In deployment, tree-sitter aims, in its documentation's words, to be fast enough to parse on every keystroke and robust
enough to provide useful results in the presence of syntax errors, representing unrecognized text as error nodes and
inserted recovery tokens as zero-width missing nodes~\citep{treesitter}. The property claimed is usefulness, stated
operationally; no soundness property of the re-anchoring position is asserted, which is a practical target the theorem
here is meant to sharpen.

Where these lines single out a resume position at all, it is chosen by convention, validated by a score, or inherited
from history; to our knowledge, none states a property of a fixed broken-input position quantified over every prefix
repair whose scan commits through its supporting evidence, complete repairs the special case. Repair invariance is that
scoped statement, and the certificates satisfying it were already posted for a different purpose. Certified recovery
does not prove that damaged input has a correct continuation. It proves the narrower, operationally useful claim that
whenever the search returns preserved evidence, every repair whose scan commits through that evidence agrees on the
returned boundary.

\begin{table}[!htbp]
\centering
\small
\renewcommand{\arraystretch}{1.4}
\begin{tabular}{@{}>{\raggedright\arraybackslash}p{0.17\linewidth}p{0.22\linewidth}p{0.29\linewidth}p{0.24\linewidth}@{}}
\toprule
Line & Synchronizes on & Guarantee quantifies over & Failure or refusal \\
\midrule
Parser-level panic mode & designated synchronizing tokens, chosen by the designer, delimiters typically & nothing beyond
termination, and membership when a synchronizer is found, but no repair-universal boundary guarantee &
no refusal, discarding to a synchronizer or the input's end \\
\addlinespace[2.5pt]
Boullier and Jourdan & ordered correction models at the detection point, else key terminals & the chosen correction
only, by continued acceptance & when local repair fails, the parser skips to a key terminal and the scanner deletes a
character \\
\addlinespace[2.5pt]
Synchronizing words for prefix codes & an intact occurrence of a synchronizing word & every decoder state
the word is defined on & for partial decoders, possibly undefined from the error's state \\
\addlinespace[2.5pt]
This paper's certificates & a certified byte or certified window occurrence & every prefix repair whose
scan commits through the returned evidence & explicit refusal when no searched certificate lies ahead \\
\bottomrule
\end{tabular}
\renewcommand{\arraystretch}{1.0}
\caption{Three axes across the lines that re-anchor a damaged stream: what each line's search looks for in the input,
whom the resulting guarantee binds, and what happens when the search finds nothing. Each cell states only what its
paragraph argues from the sources. Panic mode attaches termination and membership but no repair-universal boundary
guarantee, the two-level scheme's acceptance tests are per-correction, and synchronizing-word guarantees quantify over
decoder states, without quantifying over repairs of a broken prefix. Panic mode, synchronizing words, and the
certificates all search the damaged input forward, while the two-level scheme instead edits at the point of detection.
The last row is the only one whose guarantee ranges over repairs, complete repairs the special case, and the only
forward search of the damaged input here that can end in a refusal.}
\label{tab:axes}
\end{table}

\section{Conclusion}

Parser-level panic mode chooses a delimiter by convention; the certificates choose a resume position by theorem. A
position is a sound recovery point when every repair whose scan commits through the certificate's evidence places a
token boundary there, complete repairs the corollary, and certified bytes and window origins supply such positions under
the stated anchor conditions. The guarantee costs the committed-prefix lemma, a definition of reaching, and short proofs
on top of the companions' posted theorems.

The procedure ships with its contract documented and its behavior pinned by fail-closed tests, and the measurements say
what the theorem buys. Every evidence-covered answer passed its executable landing check on every recovery move,
$40{,}885$ of $40{,}885$, and $90.8$\% of the procedure's first answers landed. The library's anchored complete-repair
decider, measured beside the procedure as a different-property comparator, won the bare row and lost the conventional
one, both known-clean arms holding their contractual $100$ percent. The two repair readings are proved equal at
$40{,}450$ of the $43{,}736$ answers, $92.5$ percent. At $26{,}928$ of them a completing repair, exhibited and scanned
by the harness, forces the equality, and at $13{,}522$ a resumed suffix that tokenizes whole forces it, the negative
labels, proved exact, making both sets empty there. Only $3{,}286$ remain candidates for strict inclusion.

Section~\ref{sec:exact} decides inhabitation and invariance using at most $n + 1$ representative repairs. It gives sharp
witness bounds, one nonvacuous threshold, a consensus basis of at most three repairs, and an evidence profile of every
anchor in $O(nm + n|\Sigma|)$ time, linear in the input's length for a fixed automaton. The remaining $3{,}286$ campaign
answers need only their preserved tails.

On the generated JSON row, where certificates are dense, newline-past's mean absolute overshoot is seven times certified
recovery's, $29.5$ bytes against $4.2$. On the campaign's rows where certificates are sparse, certified recovery pays
for its guarantee in distance, and where the searched certificates end it returns an explicit refusal.

\Needspace*{8\baselineskip}
\section*{Tools}

Large language models assisted with drafting, with checking citations against primary records and with reviewing the
implementation; the author verified each suggestion by derivation, against primary sources, or against the
implementation, its tests and the archived artifacts, and is responsible for all content.

\begingroup \setlength{\bibsep}{2pt} \bibliographystyle{plainnat}
\bibliography{refs}
\endgroup

\end{document}